\documentclass[a4paper]{article}
\usepackage[utf8]{inputenc}
\usepackage{microtype}
\usepackage{fullpage}

\renewcommand{\phi}{\varphi}
\usepackage[usenames,dvipsnames]{xcolor}

\usepackage[colorlinks,allcolors=blue]{hyperref}
\usepackage{amssymb,amsmath,amsthm,stmaryrd,amstext,wasysym}

\usepackage{tabularx}

\usepackage{enumerate, paralist}

\usepackage{caption,subcaption,graphicx}

\usepackage{tikz}
\usetikzlibrary{decorations.shapes,decorations.pathreplacing}
\usetikzlibrary{matrix}
\usetikzlibrary{arrows,shapes,positioning}

\DefineNamedColor{named}{DarkGrey}      {cmyk}{0,0,0,.8}
\DefineNamedColor{named}{LightGrey}     {cmyk}{0,0,0,.2}
\DefineNamedColor{named}{Grey}          {cmyk}{0,0,0,.4}
\DefineNamedColor{named}{Blue}          {cmyk}{1,1,0,0}
\DefineNamedColor{named}{DarkBlue}      {cmyk}{1,1,0,0.5}
\DefineNamedColor{named}{BrickRed}      {cmyk}{0,0.89,0.94,0.28}
\DefineNamedColor{named}{Brown}         {cmyk}{0,0.81,1,0.60}
\DefineNamedColor{named}{ForestGreen}   {cmyk}{0.91,0,0.88,0.12}
\DefineNamedColor{named}{DarkGreen}     {cmyk}{0.70,0,0.70,0.60}
\DefineNamedColor{named}{Purple}        {cmyk}{0.45,0.86,0,0}
\DefineNamedColor{named}{Red}           {cmyk}{0,1,1,0}
\DefineNamedColor{named}{LightRed}      {cmyk}{0,.5,.5,0}
\DefineNamedColor{named}{DarkRed}       {cmyk}{0,1,1,0.5}
\DefineNamedColor{named}{Black}         {cmyk}{0,0,0,1}
\DefineNamedColor{named}{RoyalBlue}     {cmyk}{1,0.50,0,0}
\DefineNamedColor{named}{LightBlue}     {cmyk}{.5,0.20,0,0}

\tikzset{lab/.style={fill=none, draw=none}}

\tikzset{cluster/.style={color=Red, fill=LightRed, dashed, scale=1,
    cloud, draw, cloud puffs=10, cloud puff arc=120, aspect=2, inner
    ysep=1em}}

\tikzset{cliq/.style={color=Red, fill=LightRed, scale=1, draw, inner
    ysep=1em}}

\tikzset{indep/.style={color=Red, fill=white, dashed, scale=1, draw,
    inner ysep=1em}}

\newtheorem{theorem}{Theorem}
\newtheorem{lemma}{Lemma}[section]
\newtheorem{claim}[lemma]{Claim}

\newtheorem{corollary}[lemma]{Corollary}
\newtheorem{proposition}[lemma]{Proposition}

\newtheorem{redrule}{Rule}
\newtheorem{observation}[lemma]{Observation}

\theoremstyle{definition}
\newtheorem{definition}[lemma]{Definition}

\theoremstyle{remark}

\newcommand{\defproblem}[3]{% PGD Version
  \hfill\\\smallskip\noindent%
  \begin{tabularx}{\textwidth}{|l X|}%
    \hline%
    \multicolumn{2}{|l|}{\pname{#1}}\\%
    \textbf{Input:}&#2\\%
    \textbf{Question:}&#3\\\hline%
  \end{tabularx}%
  \smallskip%
}%

\newcommand{\pname}[1]{\textnormal{\textsc{#1}}}

\newcommand{\cclass}[1]{\textnormal{\textsf{#1}}}

\newcommand{\NP}{\cclass{NP}}

\newcommand{\TE}{\pname{Threshold Editing}}
\newcommand{\STE}{\pname{Split Threshold Editing}}
\newcommand{\TC}{\pname{Threshold Completion}}
\newcommand{\TD}{\pname{Threshold Deletion}}
\newcommand{\CE}{\pname{Chain Editing}}
\newcommand{\CCE}{\pname{Cobipartite Chordal Editing}}
\newcommand{\BCE}{\pname{Bipartite Chain Editing}}

\newcommand{\modulator}{threshold-modulator}
\newcommand{\cmodulator}{chain-modulator}

\newcommand{\I}{\mathcal{I}}
\newcommand{\C}{\mathcal{C}}
\newcommand{\V}{\mathcal{V}}

\DeclareMathOperator{\poly}{poly} % poly(k)
\DeclareMathOperator{\tc}{tc} % twin class
\DeclareMathOperator{\level}{lev} % level in a threshold partition 
\DeclareMathOperator{\lev}{lev} % level in a threshold partition 
\DeclareMathOperator{\ttc}{ttc} % true twin class
\DeclareMathOperator{\ftc}{ftc} % false twin class
\DeclareMathOperator{\en}{en}    % editing number

\DeclareMathOperator{\lex}{lex}    % lexicographic order

\newcommand{\lexlt}{<_{\lex}}

\newcommand{\order}{\pi_{\phi}} % the natural order of reduction vertices

\title{On the Threshold of Intractability}

\author{%
  Pål Grønås Drange\thanks{Dept.~Informatics, Univ.~Bergen, Norway,
    \texttt{pal.drange@ii.uib.no}, \texttt{markus.dregi@ii.uib.no},
    \texttt{daniello@ii.uib.no}}
  \and Markus Sortland Dregi\footnotemark[1]
  \and Daniel Lokshtanov\footnotemark[1]
  \and Blair D.\ Sullivan\thanks{Dept.~Computer Science, North Carolina State
    University, Raleigh, NC, USA, \texttt{blair\_sullivan@ncsu.edu}}
}

\begin{document}
\maketitle
\begin{abstract}
  We study the computational complexity of the graph modification problems
  \pname{Threshold Editing} and \pname{Chain Editing}, adding and deleting as
  few edges as possible to transform the input into a threshold (or chain)
  graph.  In this article, we show that both problems are \NP-hard, resolving a
  conjecture by Natanzon, Shamir, and Sharan~(Discrete Applied Mathematics,
  113(1):109--128, 2001).
  On the positive side, we show the problem admits a quadratic vertex kernel.
  Furthermore, we give a subexponential time parameterized algorithm solving
  \pname{Threshold Editing} in $2^{O(\sqrt k \log k)} + \poly(n)$ time, making
  it one of relatively few natural problems in this complexity class on general
  graphs.
  These results are of broader interest to the field of social network analysis,
  where recent work of Brandes~(ISAAC, 2014) posits that the minimum edit
  distance to a threshold graph gives a good measure of consistency for node
  centralities.
  Finally, we show that all our positive results extend to the related problem of
  \pname{Chain Editing}, as well as the completion and deletion variants of both
  problems.
\end{abstract}

% Section 1
\section{Introduction}
\label{sec:introduction}

In this paper we study the computational complexity of two edge modification
problems, namely editing to threshold graphs and editing to chain graphs.  Graph
modification problems ask whether a given graph~$G$ can be transformed to have a
certain property using a small number of edits (such as deleting/adding vertices
or edges), and have been the subject of significant previous
work~\cite{shamir2004cluster, chen2003computing, damaschke2006parameterized,
  dehne2006cluster, nastos2013familial}.

In the \pname{Threshold Editing} problem, we are given as input an $n$-vertex
graph $G = (V,E)$ and a non-negative integer $k$.  The objective is to find a
set $F$ of at most~$k$ pairs of vertices such that $G$ minus any edges in $F$
plus all non-edges in $F$ is a threshold graph.  A graph is a \emph{threshold
  graph} if it can be constructed from the empty graph by repeatedly adding
either an isolated vertex or a universal vertex~\cite{brandstadt1999graph}.

\defproblem{Threshold Editing}
{A graph $G$ and a non-negative integer $k$}
{Is there a set $F \subseteq [V]^2$ of size at most~$k$ such that $G \triangle F$ is a threshold
  graph.}
The computational complexity of \pname{Threshold Editing} has repeatedly been
stated as open, starting from Natanzon et al.~\cite{natanzon2001complexity}, and
then more recently by Burzyn et al.~\cite{burzyn2006np}, and again very recently
by Liu, Wang, Guo and Chen~\cite{liu2012complexity}.  We resolve this by showing
that the problem is indeed \cclass{NP}-hard.

\begin{theorem}
  \label{thm:te-np-complete}
  \TE{} is \cclass{NP}-complete, even on split graphs.
\end{theorem}

Graph editing problems are well-motivated by problems arising in the applied
sciences, where we often have a predicted model from domain knowledge, but
observed data fails to fit this model exactly.  In this setting, edge
modification corresponds to correcting false positives (and/or false negatives)
to obtain data that is consistent with the model.  \TE{} has specifically been
of recent interest in the social sciences, where Brandes et al.\ are using
distance to threshold graphs in work on axiomatization of centrality
measures~\cite{brandes2014invited,schoch2015stars}.  More generally, editing to
threshold graphs and their close relatives \emph{chain graphs} arises in the
study of sparse matrix multiplications~\cite{yannakakis1981computing}.  Chain
graphs are the bipartite analogue of threshold graphs (see
Definition~\ref{def:chain}), and here we also establish hardness of \CE{}.

\begin{theorem}
  \label{thm:ce-np-complete}
  \pname{Chain Editing} is \cclass{NP}-complete, even on bipartite graphs.
\end{theorem}

Our final complexity result is for \pname{Chordal Editing} --- a problem whose
\NP-hardness is well-known and widely used.  This result also follows from our
techniques, and as the authors were unable to find a proof in the literature, we
include this argument for the sake of completeness.

Having settled the complexity of these problems, we turn to studying ways of
dealing with their intractability.  Cai's theorem~\cite{cai1996fixed} shows that
\TE{} and \CE{} are \emph{fixed parameter tractable}, i.e., solvable in~$f(k)
\cdot \poly(n)$ time where~$k$ is the edit distance from the desired model
(graph class); However, the lower bounds we prove when showing \NP-hardness are
on the order of $2^{o(\sqrt k)}$ under ETH, and thus leave a gap.  We show that
it is in fact the lower bound which is tight (up to logarithmic factors in the
exponent) by giving a subexponential time algorithm for both problems.

\begin{theorem}
  \label{thm:te-subept}
  \TE{} and \CE{} admit ~$2^{O(\sqrt k \log k)} + \poly(n)$ subexponential time
  algorithms.
\end{theorem}

Since our results also hold for the \emph{completion} and \emph{deletion}
variants of both problems (when $F$ is restricted to be a set of non-edges or
edges, respectively), this also answers a question of Liu et
al.~\cite{liu2015edge} by giving a subexponential time algorithm for
\pname{Chain Edge Deletion}.

A crucial first step in our algorithms is to preprocess the instance, reducing
to a kernel of size polynomial in the parameter.  We give quadratic kernels for
all three variants (of both \TE{} and \CE{}).

\begin{theorem}
  \label{thm:te-tc-td-kernel}
  \TE{}, \TC{}, and \TD{} admit polynomial kernels with~$O(k^2)$ vertices.
\end{theorem}

This answers (affirmatively) a recent question of Liu, Wang and
Guo~\cite{liu2014overview}---whether the previously known kernel, which has
$O(k^3)$ vertices, for \pname{Threshold Completion} (equivalently
\pname{Threshold Deletion}) can be improved.

% Section 2
\section{Preliminaries}
\label{sec:preliminaries}

%% Begin figure: C4, P4, 2K2
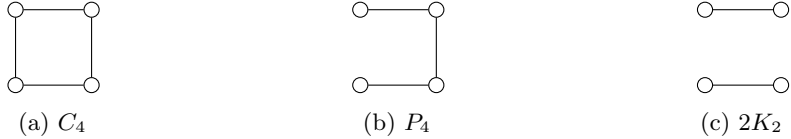
\begin{figure}[t]
  \centering
  \begin{subfigure}[t]{0.25\textwidth}
    \centering
    \begin{tikzpicture}[every node/.style={circle, draw, scale=.6},
      scale=1]
      \node (1) at (0,0) {};
      \node (2) at (1,0) {};
      \node (3) at (0,1) {};
      \node (4) at (1,1) {};
      
      \draw (1) -- (2);
      \draw (2) -- (4);
      \draw (3) -- (4);
      \draw (1) -- (3);
    \end{tikzpicture}
    \caption{$C_4$}
  \end{subfigure}
  \hspace{.02\textwidth}
  \begin{subfigure}[t]{0.25\textwidth}
    \centering
    \begin{tikzpicture}[every node/.style={circle, draw, scale=.6},
      scale=1]
      \node (1) at (0,0) {};
      \node (2) at (1,0) {};
      \node (3) at (0,1) {};
      \node (4) at (1,1) {};
      
      \draw (1) -- (2);
      \draw (2) -- (4);
      \draw (3) -- (4);
    \end{tikzpicture}
    \caption{$P_4$}
  \end{subfigure}
  \hspace{.02\textwidth}
  \begin{subfigure}[t]{0.25\textwidth}
    \centering
    \begin{tikzpicture}[every node/.style={circle, draw, scale=.6},
      scale=1]
      \node (1) at (0,0) {};
      \node (2) at (1,0) {};
      \node (3) at (0,1) {};
      \node (4) at (1,1) {};
      
      \draw (1) -- (2);

      \draw (3) -- (4);

    \end{tikzpicture}
    \caption{$2K_2$}
  \end{subfigure}
  \caption{\emph{Threshold graphs} are $\{C_4, P_4, 2K_2\}$-free.
    \emph{Chain graphs} are bipartite graphs that are $2K_2$-free.}
  \label{fig:forbidden-graphs}
\end{figure}

%% GRAPH STUFF

\paragraph{Graphs.}
We will consider only undirected simple finite graphs.  For a graph~$G$, let
$V(G)$ and~$E(G)$ denote the vertex set and the edge set of~$G$, respectively.
For a vertex $v \in V(G)$, by $N_G(v)$ we denote the open neighborhood of~$v$,
i.e.  $N_G(v)=\{u \in V(G) \mid uv \in E(G)\}$.  The closed neighborhood of~$v$,
denoted by~$N_G[v]$, is defined as $N_G(v)\cup \{v\}$.  These notions are
extended to subsets of vertices as follows: $N_G[X]=\bigcup_{v\in X} N_G[v]$ and
$N_G(X)=N_G[X]\setminus X$.  We omit the subscript whenever~$G$ is clear from
context.

When $U \subseteq V(G)$ is a subset of vertices of~$G$, we write~$G[U]$ to denote the
\emph{induced subgraph} of~$G$, i.e., the graph $G' = (U,E_U)$ where~$E_U$
is~$E(G)$ restricted to~$U$.  The degree of a vertex $v \in V(G)$, denoted
$\deg_G(v)$, is the number of vertices it is adjacent to, i.e., $\deg_G(v) =
|N_G(v)|$.  We denote by~$\Delta(G)$ the maximum degree in the graph, i.e.,
$\Delta(G) = \max_{v \in V(G)}\deg(v)$.  For a set~$A$, we write $\binom{A}{2}$
to denote the set of unordered pairs of elements of~$A$; thus $E(G) \subseteq
\binom{V(G)}{2}$.  By $\overline{G}$ we denote the \emph{complement} of
graph~$G$, i.e., $V(\overline{G})=V(G)$ and $E(\overline{G})=[V(G)]^2\setminus E(G)$.

For two sets~$A$ and~$B$ we define the \emph{symmetric difference} of~$A$
and~$B$, denoted $A \triangle B$ as the set $(A \setminus B) \cup (B \setminus
A)$.  For a graph $G = (V,E)$ and $F \subseteq [V]^2$ we define~$G \triangle F$
as the graph $(V, E \triangle F)$.

For a graph~$G$ and a vertex~$v$ we define the \emph{true twin class} of~$v$,
denoted $\ttc(v)$ as the set
$\ttc(v) = \left\{ u \in V(G) \mid N[u] = N[v] \right\}$.  Similarly, we define
the \emph{false twin class} of~$v$, denoted~$\ftc(v)$ as the set
$\ftc(v) = \left\{ u \in V(G) \mid N(u) = N(v) \right\}$.  Observe that either
$\ttc(v) = \left\{ v \right\}$ or $\ftc(v) = \left\{ v \right\}$.  From this we
define the \emph{twin class} of~$v$, denoted~$\tc(v)$ as~$\ttc(v)$ if
% 
% |ttc(v)| > | ftc(v)|
% 
$\left| \ttc(v) \right| > \left| \ftc(v) \right|$
and~$\ftc(v)$ otherwise.

%% SPLIT AND THRESHOLD STUFF

\paragraph{Split and threshold graphs.}
A split graph is a graph~$G = (V,E)$ whose vertex set can be partitioned into
two sets~$C$ and~$I$ such that~$G[C]$ is a complete graph and~$G[I]$ is
edgeless, i.e.,~$C$ is a clique and~$I$ an independent
set~\cite{brandstadt1999graph}.
For a split graph~$G$ we say that a partition $(C,I)$ of $V(G)$ forms a
\emph{split partition} of~$G$ if~$G[C]$ induces a clique and~$G[I]$ an
independent set.  A split partition $(C,I)$ is called a \emph{complete} split
partition if for every vertex $v \in I$, $N(v) = C$.  If $G$ admits a complete
split partition, we say that $G$ is a complete split graph.

We now give two useful characterizations of threshold graphs:

\begin{proposition}[\cite{mahadev1995threshold}]
  \label{prop:threshold-nested-neighborhoods}
  A graph~$G$ is a threshold graph if and only if~$G$ has a split
  partition~$(C,I)$ such that the neighborhoods of the vertices in~$I$ are
  nested, i.e., for every pair of vertices $v$ and $u$, either $N(v) \subseteq
  N[u]$ or $N(u) \subseteq N[v]$.
\end{proposition}

\begin{proposition}[\cite{brandstadt1999graph}]
  \label{prop:c4p42k2-free}
  A graph~$G$ is a threshold graph if and only if~$G$ does not have
  a~$C_4$,~$P_4$ nor a~$2K_2$ as an induced subgraph.  Thus, the threshold
  graphs are exactly the \emph{$\{ C_4, P_4, 2K_2 \}$-free graphs} (see
  Figure~\ref{fig:forbidden-graphs}).
\end{proposition}

%
% Figure of threshold partition
%
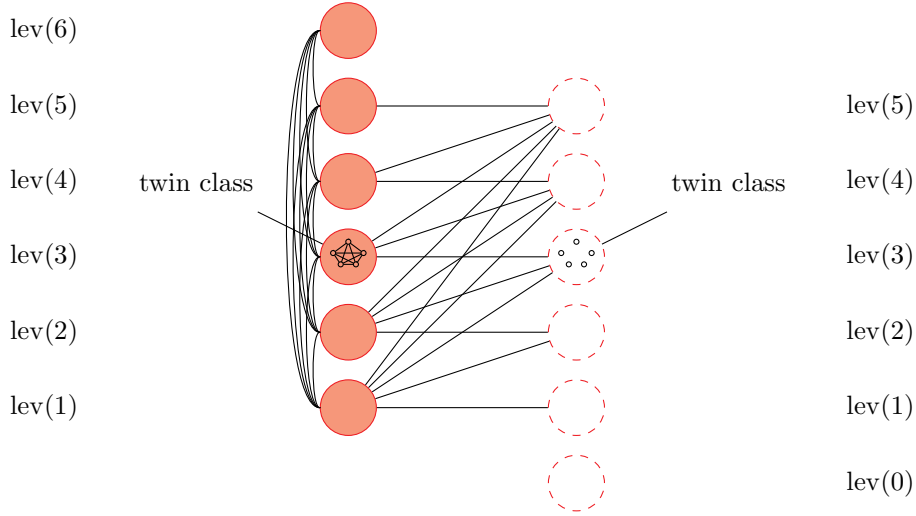
\begin{figure}[t]
  \centering
  \begin{tikzpicture}[every node/.style={circle,draw,scale=1},scale=1]
    
    \node[cliq] (c0) at (0,6) {};
    \node[cliq] (c1) at (0,5) {};
    \node[cliq] (c2) at (0,4) {};
    \node[cliq] (c3) at (0,3) {};
    \node[cliq] (c4) at (0,2) {};
    \node[cliq] (c5) at (0,1) {};
    
    \node[indep] (i1) at (3,5) {};
    \node[indep] (i2) at (3,4) {};
    \node[indep] (i3) at (3,3) {};
    \node[indep] (i4) at (3,2) {};
    \node[indep] (i5) at (3,1) {};
    \node[indep] (i6) at (3,0) {};

    %% between c and i
    \foreach \k in {1,...,5} {
      \foreach \j in {\k,...,5} {
        \draw (c\j) -- (i\k);
      }
    }
    
    %% clique
    \foreach \k in {0,...,5} {
      \foreach \j in {\k,...,5} {
        \draw (c\j) to[looseness=.3,out=180,in=180] (c\k);
      }
    }
    
    \node[lab] (il) at (5,4) {twin class};
    \draw (il) -> (i3);

    \node[scale=0.2] (ix1) at (2.90,2.90) {};
    \node[scale=0.2] (ix2) at (3.10,2.90) {};
    \node[scale=0.2] (ix3) at (2.80,3.05) {};
    \node[scale=0.2] (ix4) at (3.20,3.05) {};
    \node[scale=0.2] (ix5) at (3.00,3.20) {};
    
    \node[lab] (cl) at (-2,4) {twin class};
    \draw (cl) -> (c3);
    \node[scale=0.2] (cx1) at (-0.1,2.90) {};
    \node[scale=0.2] (cx2) at (0.10,2.90) {};
    \node[scale=0.2] (cx3) at (-0.2,3.05) {};
    \node[scale=0.2] (cx4) at (0.20,3.05) {};
    \node[scale=0.2] (cx5) at (0.00,3.20) {};
    
    \foreach \k in {1,...,5} {
      \foreach \j in {\k,...,5} {
        \draw (cx\j) -- (cx\k);
      }
    }

    \foreach \k in {1,...,6} {
      \node[lab] (labc\k) at (-4,\k) {$\lev(\k)$};
    }

    \foreach \k in {0,...,5} {
      \node[lab] (labi\k) at ( 7,\k) {$\lev(\k)$};
    }

  \end{tikzpicture}
  \caption{A threshold partition---the left hand side is the clique and the right
    hand side is an independent set, each bag contains a twin class.  All bags
    are non-empty, otherwise two twin classes on the opposite side would
    collapse into one, except possibly the two extremal bags.}
  \label{fig:threshold-partition}
\end{figure}

\begin{definition}[Threshold partition, $\level(v)$]
  \label{def:threshold-partition}
  We say that $(\mathcal{C},\mathcal{I})=(\langle C_1,\dots,C_t\rangle,\langle I_1,\dots, I_t\rangle)$ forms a
  \emph{threshold partition} of $G$ if the following holds (see
  Figure~\ref{fig:threshold-partition} for an illustration):
  \begin{itemize}
  \item $(C,I)$ is a split partition of $G$, where $C = \bigcup_{i \leq t}C_i$ and
    $I = \bigcup_{i \leq t}I_i$,
  \item $C_i$ and $I_i$ are twin classes in $G$ for every $i$
  \item $N[C_j] \subset N[C_i]$ and $N(I_i) \subset N(I_j)$ for every $i < j$.
  \item Finally, we demand that for every $i \leq t$, $(C_i, I_{\geq i})$ form a
    complete split partition of the graph induced by $C_i \cup I_{\geq i}$.
  \end{itemize}
  We furthermore define, for every vertex $v$ in $G$, $\level(v)$ as the
  number~$i$ such that $v \in C_i \cup I_i$ and we denote each level $L_i = C_i
  \cup I_i$.
\end{definition}

In a threshold decomposition we will refer to $C_i$ for every $i$ as a
\emph{clique fragment} and $I_i$ as a \emph{independent fragment}. Furthermore,
we will refer to a vertex in $\cup \mathcal{C}$ as a \emph{clique vertex} and a
vertex in $\cup \mathcal{I}$ as an \emph{independent vertex}.

\begin{proposition}[Threshold decomposition]
  \label{prop:threshold-decomposition}
  A graph~$G$ is a threshold graph if and only if~$G$ admits a threshold
  partition.
\end{proposition}
\begin{proof}
  Suppose that~$G$ is a threshold graph and therefore admits a nested ordering
  of the neighborhoods of vertices of each side~\cite{hammer1981threshold}.
  We show that partitioning the graph into partitions depending only on their
  degree yields the levels of a threshold partition.  The clique side is
  naturally defined as the maximal set of highest degree vertices that form a
  clique.  Suppose now for contradiction that this did not constitute a
  threshold partition.  By definitions, every level consists of twin classes,
  and also, for two twin classes $I_i$ and $I_j$, since their neighborhoods are
  nested in the threshold graph, their neighborhoods are nested in the threshold
  partition as well.  So what is left to verify is that $(C_i , I_{\geq i})$ is
  a complete split partition of $G[C_i \cup I_{\geq i}]$.  But that follows
  directly from the assumption that $G$ admitted a nested ordering and $C_i$ is
  a true twin class.

  For the reverse direction, suppose $G$ admits a threshold partition
  $(\mathcal{C},\mathcal{I})$.  Consider any four connected vertices $a,b,c,d$.
  We will show that they can not form any of the induced obstructions (see
  Figure~\ref{fig:forbidden-graphs}).  For the $2K_2$ and $C_4$, it is easy to
  see that at most two of the vertices can be in the clique part of the
  decomposition---and they must be adjacent since it is a clique---and hence
  there must be an edge in the independent set part of the decomposition, which
  contradicts the assumption that $\mathcal{C},\mathcal{I}$ was a threshold
  partition.  So suppose now that $a,b,c,d$ forms a $P_4$.  Again with the same
  reasoning as above, the middle edge $b,c$ must be contained in the clique
  part, hence $a$ and $d$ must be in the independent set part.  But since the
  neighborhoods of $a$ and $d$ should be nested, they cannot have a private
  neighbor each, hence either $ac$ or $bd$ must be an edge, which contradicts
  the assumption that $a,b,c,d$ induced a $P_4$.  This concludes the proof.
\end{proof}

\begin{lemma}
  \label{lemma:no-switching-of-nestedness}
  For every instance~$(G,k)$ of \TE{} or \TC{} it holds that there exists an
  optimal solution~$F$ such that for every pair of vertices~$u,v \in V(G)$,
  if~$N_G(u) \subseteq N_G[v]$ then~$N_{G \triangle F}(u) \subseteq N_{G \triangle F}[v]$.
\end{lemma}
\begin{proof}
  Let us define, for any editing set $F$ and two vertices $u$ and $v$, the set
  \[
  F_{v \leftrightarrow u} = \{e \mid e' \in F \text{ and $e$ is $e'$ with $u$ and $v$ switched}
  \}.
  \]
  Suppose $F$ is an optimal solution for which the above statement does not
  hold.  Then $N_G(u) \subseteq N_G[v]$ and $N_{G \triangle F}(v) \subseteq N_{G
    \triangle F}[u]$ (see
  Proposition~\ref{prop:threshold-nested-neighborhoods}).  But then it is easy
  to see that we can flip edges in an ordering such that at some point, say
  after flipping $F^0$, $u$ and $v$ are twins in this intermediate graph $G
  \triangle F^0$.  Let $F^1 = F \setminus F^0$.  It is clear that for $G' = G
  \triangle (F^0 \cup F^1_{v \leftrightarrow u})$, $N_{G'}(u) \subseteq N_{G'}[v]$.
  Since $|F| \geq |F^0 \cup F^1_{v \leftrightarrow u}|$, the claim holds.
\end{proof}

\paragraph{Chain graphs.}

Chain graphs are the bipartite graphs whose neighborhoods of the vertices on one
of the sides form an inclusion chain.  It follows that the neighborhoods on the
opposite side form an inclusion chain as well.  If this is the case, we say that
the neighborhoods are \emph{nested}.  The relation to threshold graphs is
obvious, see Figure~\ref{fig:chain-threshold} for a comparison.  The problem of
completing edges to obtain a chain graph was introduced by
Golumbic~\cite{golumbic-book} and later studied by
Yannakakis~\cite{yannakakis1981computing}, Feder, Mannila and
Terzi~\cite{feder2009approximating} and finally by Fomin and
Villanger~\cite{fomin2012subexponential} who showed that \pname{Chain
  Completion} when given a bipartite graph whose bipartition must be respected
is solvable in subexponential time.

\begin{definition}[Chain graph]
  \label{def:chain}
  A bipartite graph $G = (A,B,E)$ is a chain graph if there is an ordering of
  the vertices of $A$, $a_1, a_2, \dots, a_{|A|}$ such that $N(a_1) \subseteq
  N(a_2) \subseteq \cdots \subseteq N(a_{|A|})$.
\end{definition}

From the following proposition, it follows that chain graphs are characterized
by a finite set of forbidden induced subgraphs and hence are subject to Cai's
theorem~\cite{cai1996fixed}.

\begin{proposition}[\cite{brandstadt1999graph}]
  \label{prop:chain-defs}
  Let $G$ be a graph.  The following are equivalent:
  \begin{itemize}
  \item $G$ is a chain graph.
  \item $G$ is bipartite and $2K_2$-free.
  \item $G$ is $\{2K_2, C_3, C_5\}$-free.
  \item $G$ can be constructed from a threshold graph by removing all the edges
    in the clique partition.
  \end{itemize}
\end{proposition}

Since they have the same structure as threshold graphs, it is natural to talk
about a \emph{chain decomposition}, $(\mathcal{A},\mathcal{B})$ of a bipartite
graph~$G$ with bipartition $(A,B)$.  We say that $(\mathcal{A},\mathcal{B})$ is
a chain decomposition for a chain graph~$G$ if and only if
$(\mathcal{A},\mathcal{B})$ is a threshold decomposition for the corresponding
threshold graph~$G'$ where~$A$ is made into a clique.

%
%
% CHAIN & THRESHOLD
%
%
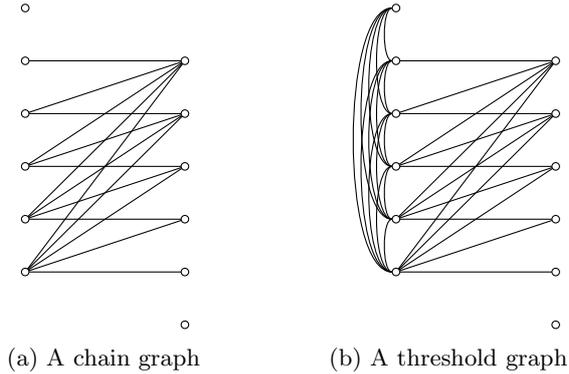
\begin{figure}[t]
  \centering
  \begin{subfigure}[t]{0.25\textwidth}
    \centering
    \begin{tikzpicture}[every node/.style={circle, draw, scale=.3},
      scale=.7]
      \node (c1) at (0,0) {};
      \node (c2) at (0,1) {};
      \node (c3) at (0,2) {};
      \node (c4) at (0,3) {};
      \node (c5) at (0,4) {};
      \node (c6) at (0,5) {};

      \node (i0) at (3,-1) {};
      \node (i1) at (3,0) {};
      \node (i2) at (3,1) {};
      \node (i3) at (3,2) {};
      \node (i4) at (3,3) {};
      \node (i5) at (3,4) {};

      \draw (c1) -- (i1);
      \draw (c1) -- (i2);
      \draw (c1) -- (i3);
      \draw (c1) -- (i4);
      \draw (c1) -- (i5);

      \draw (c2) -- (i2);
      \draw (c2) -- (i3);
      \draw (c2) -- (i4);
      \draw (c2) -- (i5);

      \draw (c3) -- (i3);
      \draw (c3) -- (i4);
      \draw (c3) -- (i5);

      \draw (c4) -- (i4);
      \draw (c4) -- (i5);
      
      \draw (c5) -- (i5);
    \end{tikzpicture}
    \caption{A chain graph}
  \end{subfigure}
  \hspace{.02\textwidth}
  \begin{subfigure}[t]{0.25\textwidth}
    \centering
    \begin{tikzpicture}[every node/.style={circle, draw, scale=.3},
      scale=.7]
      \node (c1) at (0,0) {};
      \node (c2) at (0,1) {};
      \node (c3) at (0,2) {};
      \node (c4) at (0,3) {};
      \node (c5) at (0,4) {};
      \node (c6) at (0,5) {};

      \node (i0) at (3,-1) {};
      \node (i1) at (3,0) {};
      \node (i2) at (3,1) {};
      \node (i3) at (3,2) {};
      \node (i4) at (3,3) {};
      \node (i5) at (3,4) {};

      \draw (c1) -- (i1);
      \draw (c1) -- (i2);
      \draw (c1) -- (i3);
      \draw (c1) -- (i4);
      \draw (c1) -- (i5);

      \draw (c2) -- (i2);
      \draw (c2) -- (i3);
      \draw (c2) -- (i4);
      \draw (c2) -- (i5);

      \draw (c3) -- (i3);
      \draw (c3) -- (i4);
      \draw (c3) -- (i5);

      \draw (c4) -- (i4);
      \draw (c4) -- (i5);
      
      \draw (c5) -- (i5);

      \draw (c1) to[out=180,in=180,looseness=0.5] (c2);
      \draw (c1) to[out=180,in=180,looseness=0.5] (c3);
      \draw (c1) to[out=180,in=180,looseness=0.5] (c4);
      \draw (c1) to[out=180,in=180,looseness=0.5] (c5);
      \draw (c1) to[out=180,in=180,looseness=0.5] (c6);

      \draw (c2) to[out=180,in=180,looseness=0.5] (c3);
      \draw (c2) to[out=180,in=180,looseness=0.5] (c4);
      \draw (c2) to[out=180,in=180,looseness=0.5] (c5);
      \draw (c2) to[out=180,in=180,looseness=0.5] (c6);

      \draw (c3) to[out=180,in=180,looseness=0.5] (c4);
      \draw (c3) to[out=180,in=180,looseness=0.5] (c5);
      \draw (c3) to[out=180,in=180,looseness=0.5] (c6);
      
      \draw (c4) to[out=180,in=180,looseness=0.5] (c5);
      \draw (c4) to[out=180,in=180,looseness=0.5] (c6);
      
      \draw (c5) to[out=180,in=180,looseness=0.5] (c6);

    \end{tikzpicture}
    \caption{A threshold graph}
  \end{subfigure}
  \caption{Illustration of the similarities between chain and threshold
    graphs.  Note that the nodes drawn can be replaced by twin classes
    of any size, even empty.  However, if on one side of a level there
    is an empty class, the other two levels on the opposite side will
    collapse to a twin class.  See
    Proposition~\ref{prop:threshold-decomposition}.}
  \label{fig:chain-threshold}
\end{figure}

%
% PARAMETERIZED STUFF
%

\paragraph{Parameterized complexity.}
The running time of an algorithm in classical complexity analysis is described
as a function of the length of the input.  To refine the analysis of
computationally hard problems, especially \NP-hard problems, parameterized
complexity introduced the notion of an extra ``parameter''---an additional part
of a problem instance used to measure the problem complexity when the parameter
is taken into consideration.  To simplify the notation, here we consider inputs
to problems to be of the form~$(G,k)$---a pair consisting of a graph~$G$ and a
nonnegative integer~$k$.  We will say that a problem is \emph{fixed parameter
  tractable} whenever there is an algorithm solving the problem in time $f(k)
\cdot \poly(|G|)$, where~$f$ is any function, and $\poly \colon \mathbb{N} \to
\mathbb{N}$ any polynomial function.  In the case when $f(k) = 2^{o(k)}$ we say
that the algorithm is a subexponential parameterized algorithm.  When a problem
$\Pi \subseteq \mathcal{G} \times \mathbb{N}$ is fixed-parameter tractable,
where~$\mathcal{G}$ is the class of all graphs, we say that~$\Pi$ belongs to the
complexity class \cclass{FPT}.  For a more rigorous introduction to
parameterized complexity we refer to the book of Flum and
Grohe~\cite{flum2006parameterized}.

Given a parameterized problem $\Pi$, we say two instances $(G,k)$ and $(G',k')$
are {\em equivalent} if $(G,k)\in \Pi$ if and only if $(G',k')\in \Pi$.  A
\emph{kernelization algorithm} (or \emph{kernel}) is a polynomial-time algorithm
for a parameterized problem~$\Pi$ that takes as input a problem instance $(G,k)$
and returns an equivalent instance $(G',k')$, where both~$|G'|$ and~$k'$ are
bounded by~$f(k)$ for some function~$f$.  We then say that~$f$ is the \emph{size
  of the kernel}.  When $k' \leq k$, we say that the kernel is a \emph{proper
  kernel}.  Specifically, a proper polynomial kernelization algorithm for~$\Pi$
is a polynomial time algorithm which takes as input an instance $(G,k)$ and
returns an equivalent instance~$(G',k')$ with $k' \leq k$ and $|G'| \leq p(k)$
for some polynomial function~$p$.

\begin{definition}[Laminar set system, \cite{drange2014apolynomial}]
  \label{def:laminar}
  A set system $\mathcal{F}\subseteq 2^U$ over a ground set $U$ is called \emph{laminar}
  if for every $X_1$ and $X_2$ in $\mathcal{F}$ with $x_1 \in X_1 \setminus X_2$
  and $x_2 \in X_2 \setminus X_1$, there is no $Y \in \mathcal{F}$ with $\{x_1,
  x_2\} \subseteq Y$.
\end{definition}

An equivalent way of looking at a laminar set system $\mathcal{F}$ is that every
two sets $X_1$ and $X_2$ in $\mathcal{F}$ are either disjoint or nested, that
is, for every $X_1, X_2 \in \mathcal{F}$ either $X_1 \cap X_2 = \emptyset$, or
$X_1 \subseteq X_2$ or $X_2 \subseteq X_1$.

\begin{lemma}[\cite{drange2014apolynomial}]
  \label{lem:laminar-bounded}
  Let~$\mathcal{F}$ be a laminar set system over a finite ground set~$U$.  Then
  the cardinality of~$\mathcal{F}$ is at most~$|U| + 1$.
\end{lemma}

% Section 3
\section{Hardness}
\label{sec:hardness}

In this section we show that \pname{Threshold Editing} is \NP-complete.
Recalling (see Figure~\ref{fig:chain-threshold}) that chain graphs are bipartite
graphs with structure very similar to that of threshold graphs, it should not be
surprising that we obtain as a corollary that \pname{Chain Editing} is
\NP-complete as well.

We will also conclude the section by giving a proof for the fact that
\pname{Chordal Editing} is \NP-complete; Although this has been known for a long
time (Natanzon~\cite{natanzon1999complexity}, Natanzon et
al.~\cite{natanzon2001complexity}, Sharan~\cite{sharan2002graph}),
the authors were unable to find a proof in the literature for the
\NP-completeness of \pname{Chordal Editing} and therefore include the
observation.
The problem was recently shown to be \cclass{FPT} by Cao and
Marx~\cite{cao2014chordal}, however we would like to point out that the more
general problem studied there is indeed well-known to be \NP-complete as it is a
generalized version of \pname{Chordal Vertex Deletion}.

% 
% TECHNICAL STUFF
% 

\subsection{NP-completeness of Threshold Editing}
Recall that a boolean formula~$\phi$ is in 3-CNF-SAT if it is in conjunctive normal
form and each clause has at most three variables.
Our hardness reduction is from the problem \pname{3Sat}, where we are given a
3-CNF-SAT formula~$\phi$ and asked to decide whether~$\phi$ admits a satisfying
assignment.
We will denote by~$\C_\phi$ the set of clauses, and by~$\V_\phi$ the set of variables
in a given 3-CNF-SAT formula~$\phi$.
An \emph{assignment} for a formula~$\phi$ is a function~$\alpha \colon \V_\phi \to
\{\mathtt{true},\mathtt{false}\}$.
Furthermore, we assume we have some natural lexicographical ordering~$\lexlt$ of
the clauses $c_1, \dots, c_{|\C_\phi|}$ and the same for the variables $v_1,
\dots, v_{|\V_\phi|}$, hence we may write, for some variables~$x$ and~$y$,
that~$x \lexlt y$.  To immediately get an impression of the reduction we aim
for, the construction is depicted in Figure~\ref{fig:hardness-gadget}.

% fig is [t]
\begin{figure}[t]
  \centering
  \begin{tikzpicture}[every node/.style= {circle,draw,scale=.8}, scale=.7]

    % X: 0-5
    \node (xa) at ( 0.0,  3.0) {};
    \node (xb) at ( 1.0,  3.0) {};
    \node (xf) at ( 2.0,  3.0) {};
    \node (xt) at ( 3.0,  3.0) {};
    \node (xc) at ( 4.0,  3.0) {};
    \node (xd) at ( 5.0,  3.0) {};
    
    \node[draw=none] (xat) at ( 0.0,  4.0) {$v^x_a$};
    \node[draw=none] (xbt) at ( 1.0,  4.0) {$v^x_b$};
    \node[draw=none] (xft) at ( 2.0,  4.0) {$v^x_\bot$};
    \node[draw=none] (xtt) at ( 3.0,  4.0) {$v^x_\top$};
    \node[draw=none] (xct) at ( 4.0,  4.0) {$v^x_c$};
    \node[draw=none] (xdt) at ( 5.0,  4.0) {$v^x_d$};
    
    \draw (-.3,3.3) -- (5.3,3.3) -- (5.3,2.7) -- (-.3,2.7) -- (-.3,3.3);
    
    % Y: 6-10
    \node (ya) at ( 6.0,  3.0) {};
    \node (yb) at ( 7.0,  3.0) {};
    \node (yf) at ( 8.0,  3.0) {};
    \node (yt) at ( 9.0,  3.0) {};
    \node (yc) at (10.0,  3.0) {};
    \node (yd) at (11.0,  3.0) {};

    \node[draw=none] (yat) at ( 6.0,  4.0) {$v^y_a$};
    \node[draw=none] (ybt) at ( 7.0,  4.0) {$v^y_b$};
    \node[draw=none] (yft) at ( 8.0,  4.0) {$v^y_\bot$};
    \node[draw=none] (ytt) at ( 9.0,  4.0) {$v^y_\top$};
    \node[draw=none] (yct) at (10.0,  4.0) {$v^y_c$};
    \node[draw=none] (ydt) at (11.0,  4.0) {$v^y_d$};
    
    \draw (5.7,3.3) -- (11.3,3.3) -- (11.3,2.7) -- (5.7,2.7) -- (5.7,3.3);
    
    % Z: 12-17
    \node (za) at (12.0,  3.0) {};
    \node (zb) at (13.0,  3.0) {};
    \node (zf) at (14.0,  3.0) {};
    \node (zt) at (15.0,  3.0) {};
    \node (zc) at (16.0,  3.0) {};
    \node (zd) at (17.0,  3.0) {};

    \node[draw=none] (zat) at (12.0,  4.0) {$v^z_a$};
    \node[draw=none] (zbt) at (13.0,  4.0) {$v^z_b$};
    \node[draw=none] (zft) at (14.0,  4.0) {$v^z_\bot$};
    \node[draw=none] (ztt) at (15.0,  4.0) {$v^z_\top$};
    \node[draw=none] (zct) at (16.0,  4.0) {$v^z_c$};
    \node[draw=none] (zdt) at (17.0,  4.0) {$v^z_d$};
    
    \draw (11.7,3.3) -- (17.3,3.3) -- (17.3,2.7) -- (11.7,2.7) -- (11.7,3.3);

    % CLAUSE c1
    \node (c1) at ( 5.5,  0.0) {};
    \node[draw=none] (c1t) at ( 5.0,  -1) {$v_{c_1}$};
    \node[draw=none] (c1c) at ( 5.0, -2) {$c_1 = \overline x \lor y$};

    \draw (c1) to[in=-90, out=90, looseness=0] (xb);
    \draw (c1) to[in=-90, out=90, looseness=0] (xf);
    \draw (c1) to[in=-90, out=90, looseness=0] (xd);

    \draw (c1) to[in=-90, out=90, looseness=0] (yb);
    \draw (c1) to[in=-90, out=90, looseness=0] (yt);
    \draw (c1) to[in=-90, out=90, looseness=0] (yd);

    \draw (c1) to[in=-90, out=90, looseness=0] (zb);
    \draw (c1) to[in=-90, out=90, looseness=0] (zc);
    \draw (c1) to[in=-90, out=90, looseness=0] (zd);

    % CLAUSE c2
    \node (c2) at ( 11.5,  0.0) {};
    \node[draw=none] (c2t) at ( 12.0,  -1) {$v_{c_2}$};
    \node[draw=none] (c2c) at ( 12, -2) {$c_2 =  x \lor z$};

    \draw (c2) to[in=-90, out=90, looseness=0] (xb);
    \draw (c2) to[in=-90, out=90, looseness=0] (xt);
    \draw (c2) to[in=-90, out=90, looseness=0] (xd);
                                             
    \draw (c2) to[in=-90, out=90, looseness=0] (yb);
    \draw (c2) to[in=-90, out=90, looseness=0] (yc);
    \draw (c2) to[in=-90, out=90, looseness=0] (yd);
                                             
    \draw (c2) to[in=-90, out=90, looseness=0] (zb);
    \draw (c2) to[in=-90, out=90, looseness=0] (zt);
    \draw (c2) to[in=-90, out=90, looseness=0] (zd);

  \end{tikzpicture}
  \caption{The connections of a clause and a variable.  All the vertices on the top
    (the variable vertices) belong to the clique, while the vertices on the
    bottom (the clause vertices) belong to the independent set.
    The vertices in the left part of the clique has higher degree than the
    vertices of the right part of the clique, whereas all the clause vertices
    (in the independent set) will all have the same degree, namely $3 \cdot
    |\V_\phi|$.}
  \label{fig:hardness-gadget}
\end{figure}
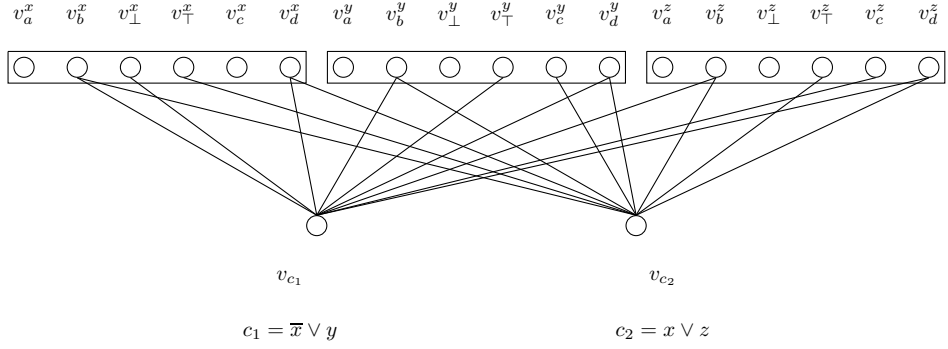

\subsubsection{Construction}
Recall that we want to form a graph~$G_\phi$ and pick an integer~$k_\phi$ so
that~$(G_\phi,k_\phi)$ is a yes-instance of \pname{Threshold Editing} if and
only if~$\phi$ is satisfiable.  We will design~$G_\phi$ to be a split graph, so
that the split partition is forced to be maintained in any threshold graph
within distance~$k_\phi$ of~$G_\phi$, where $k_\phi = |\C_\phi| \cdot (3 |\V_\phi| - 1)$.

Given~$\phi$, we first create a clique of size~$6|\V_\phi|$; To each variable $x \in
\V_\phi$, we associate six vertices of this clique, and order them in the following
manner
\[
v^x_a, v^x_b, v^x_\bot, v^x_\top, v^x_c, v^x_d .
\]
We will throughout the reduction refer to this ordering as $\order$: $\order$ is
a partial order which has 
\[
v^x_a <_{\order} v^x_b <_{\order} v^x_\top, v^x_\bot <_{\order} v^x_c <_{\order}
v^x_d,
\]
and for every two vertex $v^x_\star$ and $v^y_\star$ with $x \lexlt y$, we have
$v^x_\star <_{\order} v^y_\star$.  Observe that we do not specify which comes
first of $v^x_\top$ and $v^x_\bot$---this is the choice that will result in the
assignment~$\alpha$ for~$\phi$.

We enforce this ordering by adding~$O(k_\phi^2)$ vertices in the independent set;
Enforcing that~$v_1$ comes before~$v_2$ in the ordering is done by
adding~$k_\phi+1$ vertices in the independent set incident to all the vertices
coming before~$v_1$, including~$v_1$.  Since swapping the position of~$v_1$
and~$v_2$ would demand at least~$k_\phi+1$ edge modifications and~$k_\phi$ is
the intended budget, in any yes instance,~$v_1$ ends up before~$v_2$ in the
ordering of the clique.

We proceed adding the clause gadgets; For every clause~$c \in \C_\phi$, we add one
vertex~$v_c$ to the independent set.  Hence, the size of the independent set is
$O(|\C_\phi| + k_\phi^2)$.  For a variable~$x$ occurring in~$c$, we add an edge
between~$v_c$ and~$v^x_\bot$ if it occurs negatively, and between~$v_c$
and~$v^x_\top$ otherwise.  In addition, we make~$v_c$ incident to~$v^x_b$
and~$v^x_d$.

For a variable~$z$ which does not occur in a clause~$c$, we make $v_c$ adjacent
to $v^z_b$, $v^z_c$, and $v^z_d$.  To complete the reduction, we
add~$4(k_\phi+1)$ isolated vertices;~$k_\phi+1$ vertices to the left in the
independent set,~$k_\phi+1$ vertices to the right in the independent set,
and~$k_\phi+1$ to the left and~$k_\phi+1$ to the right in the clique.
This ensures that no vertex will move from the clique to the independent set
partition, and vice versa.

\subsubsection{Properties of the Constructed Instance}
\label{sec:properties-construction}

Before proving the Theorem~\ref{thm:te-np-complete}, and specifically
Lemma~\ref{lem:reduc-correct},
we may observe the following, which may serve as an intuition for the idea of
the reduction.  When we consider a fixed permutation of the variable gadget
vertices (the clique side), the only thing we need to determine for a clause
vertex~$v_c$, is the \emph{cut-off point}: the point in $\order$ at which the
vertex~$v_c$ will no longer have any neighbors.  Observing that no
vertex~$v^x_i$ swaps places with any other~$v^x_j$ for $i,j \in \{a,b,c,d\}$,
and that no~$v^x_\star$ changes with~$v^y_\star$ for $x,y \in \V_\phi$, consider
a fixed permutation of the variable vertices.  We charge the clause vertices
with the edits incident to the clause vertex.  Since the budget is $k_\phi =
|\C| \cdot (3 |\V_\phi| - 1)$, and every clause needs at least~$3|\V_\phi|-1$,
to obtain a solution (upcoming Lemma~\ref{lem:budget-lower}) we need to charge
every clause vertex with exactly~$3 |\V_\phi| - 1$ edits.
Figure~\ref{fig:plot} illustrates the charged cost of a clause vertex.
\begin{figure}[t]
  \centering
  \begin{tikzpicture}[every node/.style={circle,draw,scale=.7},
    scale=.6]
    
    % x axis
    \draw (0,0) -- (23,0);
    
    % y axis
    \draw (0,0) -- (0,5);
    
    \node[draw=none] (k1) at (-1,2) {$3 |\V| - 1$};
    \node[draw=none] (k2) at (-1,3) {$3 |\V|$};
    \node[draw=none] (k3) at (-1,4) {$3 |\V| + 1$};
    
    % here's the line
    \draw (1,3) -- (2,4) -- (3,3) -- (4,4) -- (5,3) -- (6,2) -- (7,3) -- (8,4) -- (9,3) -- (10,4);
    \draw (10,4) -- (11,3) -- (12,4) -- (13,3) -- (14,2) -- (15,3) -- (16,4);
    \draw (16,4) -- (17,5) -- (18,6) -- (19,5) -- (20,4) -- (21,3) --
    (22,4) -- (23,3);
    
    \draw[dashed] (-1em,3) -- (23,3);

    \draw (6,1.9) to[->,out=-90,in=90] (8,1);
    \draw (14,1.9) to[->,out=-90,in=90] (8,1);

    \node[draw=none] (tx) at (8,0.5) {Satisfying literals};

  \end{tikzpicture}
  \caption{The cost with which we charge a clause vertex depends on the cut-off
    point; The $x$-axis denotes the point in the lexicographic ordering which
    separates the vertices adjacent to the clause vertex from the vertices not
    adjacent to the clause vertex.}
  \label{fig:plot}
\end{figure}
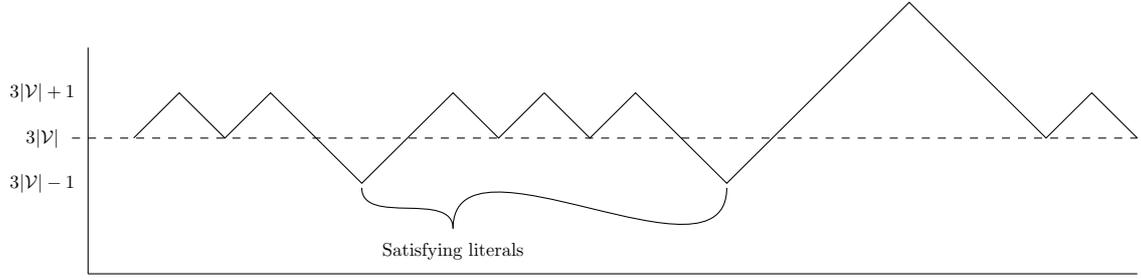

\begin{observation}
  \label{obs:g-is-split}
  The graph~$G_\phi$ resulting from the above procedure is a split graph and when
  $k_\phi = |\C| \cdot (3 |\V_\phi| - 1)$, if~$H$ is a threshold graph within
  distance~$k_\phi$ of~$G_\phi$,~$H$ must have the same clique-maximizing split
  partition as~$G_\phi$.
\end{observation}

\begin{lemma}
  \label{lem:budget-lower}
  Let $(G_\phi,k_\phi)$ be a yes instance to \TE{} constructed from a 3-CNF-SAT
  formula~$\phi$ with $|F| \leq k_\phi$ a solution.
  For any clause vertex~$v_c$, at least $3 |\V_\phi| - 1$ edges in~$F$ are incident
  to~$v_c$.
\end{lemma}
\begin{proof}
  By the properties of $\order$, we know that the only vertices we may change
  the order of are those corresponding to $v^\star_\top$ and $v^\star_\bot$.
  Pick any index in $\order$ for which we know that~$v_c$ is adjacent to all
  vertices on the left hand side and non-adjacent to all vertices on the right
  hand side.  Let~$L_c$ be the set of variables whose vertices are completely
  adjacent to~$v_c$ and~$R_c$ the corresponding set completely non-adjacent
  to~$v_c$.  By construction, $v_c$ has exactly three neighbors in each variable
  and thus these variable gadgets contribute $3 (|L_c| + |R_c|)$ to the budget.
  If $L_c \cup R_c = \V_\phi$, we are done, as~$v_c$ needs at least~$3|\V_\phi|$
  edits here.
  
  Suppose therefore that there is a variable~$x$ whose vertex~$v^x_a$ is
  adjacent to~$v_c$ and~$v^x_d$ is non-adjacent to~$v_c$.  But then we have
  already deleted the existing edge $v_cv^x_d$ and added the non-existing edge
  $v_cv^x_a$.  This immediately gives a lower bound on $3(|\V_\phi|-1)+2 =
  3|\V_\phi|-1$ edits.
\end{proof}

\subsubsection{Proof of Correctness}
\label{sec:proof-correctness}

\begin{lemma}
  \label{lem:save-satisfy}
  If there is an editing set~$F$ of size at most~$k_\phi$ for an
  instance~$(G_\phi,k_\phi)$ constructed from a 3-CNF-SAT formula~$\phi$,
  and~$|F(v_c)| = 3|\V_\phi|-1$, then the~$\lexlt$-highest vertex connected
  to~$v_c$ corresponds to a variable satisfying the clause~$c$.
\end{lemma}
\begin{proof}
  From the proof of Lemma~\ref{lem:budget-lower}, we observed that for a clause
  $c$ to be within budget, we must choose a cut-off point within a variable
  gadget, meaning that there is a variable $x$ for which $v_c$ is adjacent to
  $v^x_a$ and non-adjacent to $v^x_d$.

  We now distinguish two cases,~\textit{(i)}~$x$ is a variable occurring
  (w.l.o.g.\ positively) in~$c$ and~\textit{(ii)}~$x$ does not occur in~$c$.
  For~\textit{(i)},~$v_c$ was adjacent to~$v^x_b$,~$v^x_\top$, and~$v^x_d$.  By
  assumption, we add the edge to $v^x_a$ and delete the edge to~$v^x_d$.  But
  then we have already spent the entire budget, hence the only way this is a
  legal editing, $v^x_\top$ must come before $v^x_\bot$, and hence satisfies
  $v_c$.  See Figure~\ref{fig:hardness-gadget-edited}.

  For~\textit{(ii)} we have that~$v_c$ was adjacent to~$v^x_b$,~$v^x_c$,
  and~$v^x_d$.  Here we, again by assumption, add the edge to~$v^x_a$ and delete
  the edge to~$v^x_d$.  This alone costs two edits, so we are done.  But observe
  that these two edits alone are not enough, hence if we want to achieve the
  goal of $3 |\V_\phi| - 1$ edited edges, the cut-off index must be inside a
  variable gadget corresponding to a variable occurring in~$c$,
  i.e.~\textit{(i)}~must be the case.
\end{proof}

% fig is [t]
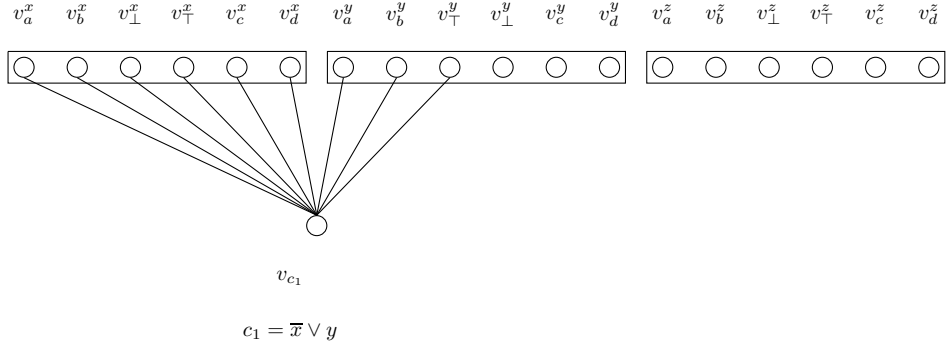
\begin{figure}[t]
  \centering
  \begin{tikzpicture}[every node/.style= {circle,draw,scale=.8}, scale=.7]
    
    % X: 0-5
    \node (xa) at ( 0.0,  3.0) {};
    \node (xb) at ( 1.0,  3.0) {};
    \node (xf) at ( 2.0,  3.0) {};
    \node (xt) at ( 3.0,  3.0) {};
    \node (xc) at ( 4.0,  3.0) {};
    \node (xd) at ( 5.0,  3.0) {};
    
    \node[draw=none] (xat) at ( 0.0,  4.0) {$v^x_a$};
    \node[draw=none] (xbt) at ( 1.0,  4.0) {$v^x_b$};
    \node[draw=none] (xft) at ( 2.0,  4.0) {$v^x_\bot$};
    \node[draw=none] (xtt) at ( 3.0,  4.0) {$v^x_\top$};
    \node[draw=none] (xct) at ( 4.0,  4.0) {$v^x_c$};
    \node[draw=none] (xdt) at ( 5.0,  4.0) {$v^x_d$};
    
    \draw (-.3,3.3) -- (5.3,3.3) -- (5.3,2.7) -- (-.3,2.7) -- (-.3,3.3);
    
    % Y: 6-10
    \node (ya) at ( 6.0,  3.0) {};
    \node (yb) at ( 7.0,  3.0) {};
    \node (yf) at ( 9.0,  3.0) {}; % OBS: These have been
    \node (yt) at ( 8.0,  3.0) {}; %      swapped
    \node (yc) at (10.0,  3.0) {};
    \node (yd) at (11.0,  3.0) {};
    
    \node[draw=none] (yat) at ( 6.0,  4.0) {$v^y_a$};
    \node[draw=none] (ybt) at ( 7.0,  4.0) {$v^y_b$};
    \node[draw=none] (yft) at ( 9.0,  4.0) {$v^y_\bot$}; % OBS: These have been 
    \node[draw=none] (ytt) at ( 8.0,  4.0) {$v^y_\top$}; %      swapped         
    \node[draw=none] (yct) at (10.0,  4.0) {$v^y_c$};
    \node[draw=none] (ydt) at (11.0,  4.0) {$v^y_d$};
    
    \draw (5.7,3.3) -- (11.3,3.3) -- (11.3,2.7) -- (5.7,2.7) -- (5.7,3.3);
    
    % Z: 12-17
    \node (za) at (12.0,  3.0) {};
    \node (zb) at (13.0,  3.0) {};
    \node (zf) at (14.0,  3.0) {};
    \node (zt) at (15.0,  3.0) {};
    \node (zc) at (16.0,  3.0) {};
    \node (zd) at (17.0,  3.0) {};
    
    \node[draw=none] (zat) at (12.0,  4.0) {$v^z_a$};
    \node[draw=none] (zbt) at (13.0,  4.0) {$v^z_b$};
    \node[draw=none] (zft) at (14.0,  4.0) {$v^z_\bot$};
    \node[draw=none] (ztt) at (15.0,  4.0) {$v^z_\top$};
    \node[draw=none] (zct) at (16.0,  4.0) {$v^z_c$};
    \node[draw=none] (zdt) at (17.0,  4.0) {$v^z_d$};
    
    \draw (11.7,3.3) -- (17.3,3.3) -- (17.3,2.7) -- (11.7,2.7) -- (11.7,3.3);

    % CLAUSE c1
    \node (c1) at ( 5.5,  0.0) {};
    \node[draw=none] (c1t) at ( 5.0,  -1) {$v_{c_1}$};
    \node[draw=none] (c1c) at ( 5.0, -2) {$c_1 = \overline x \lor y$};

    \draw (c1) to[in=-90, out=90, looseness=0] (xb);
    \draw (c1) to[in=-90, out=90, looseness=0] (xf);
    \draw (c1) to[in=-90, out=90, looseness=0] (xd);
    
    \draw (c1) to[in=-90, out=90, looseness=0] (xa); % new
    \draw (c1) to[in=-90, out=90, looseness=0] (xt); % new 
    \draw (c1) to[in=-90, out=90, looseness=0] (xc); % new

    \draw (c1) to[in=-90, out=90, looseness=0] (ya); % new
    \draw (c1) to[in=-90, out=90, looseness=0] (yb);
    \draw (c1) to[in=-90, out=90, looseness=0] (yt);
    % \draw (c1) to[in=-90, out=90, looseness=0] (yd); % deleted
    
    % \draw (c1) to[in=-90, out=90, looseness=0] (zb); % deleted
    % \draw (c1) to[in=-90, out=90, looseness=0] (zc); % deleted
    % \draw (c1) to[in=-90, out=90, looseness=0] (zd); % deleted
    
  \end{tikzpicture}
  \caption{The edited version when~$y$ satisfies~$c_1$.  We have added three edges
    to the gadget~$x$ and deleted three edges to the gadget~$z$, and added the
    edge to~$v^y_a$ and deleted the edge to~$v^y_d$, that is, we have edited
    exactly $3 \cdot 2 + 2 = 3(|\V|-1) + 2 = 3|\V| - 1$ edges incident to~$c_1$.
    Notice that if~$v^y_\bot$ was coming before~$v^y_\top$, we would have to
    choose a different variable to satisfy~$c_1$.}
  \label{fig:hardness-gadget-edited}
\end{figure}

\begin{lemma}
  \label{lem:reduc-correct}
  A 3-CNF-SAT formula~$\phi$ is satisfiable if and only if~$(G_\phi,k_\phi)$ is a yes
  instance to \pname{Threshold Editing}.
\end{lemma}
\begin{proof}[Proof of Lemma~\ref{lem:reduc-correct}]
  For the forwards direction, let~$\phi$ be a satisfiable 3-CNF-SAT formula where
  $\alpha \colon \V_\phi \to \{\mathtt{true},\mathtt{false}\}$ is any satisfying
  assignment, and~$(G_\phi,k_\phi)$ the \pname{Threshold Editing} instance as
  described above.

  % HERE WAS FIGURE

  Now, let $\alpha \colon \V_\phi \to \{\mathtt{true},\mathtt{false}\}$ be a satisfying
  assignment, and~$(G_\phi,k_\phi)$ the \pname{Threshold Editing} instance as
  described above, and let~$\pi$ be any permutation of the vertices of the
  clique side with the following properties
  \begin{itemize}
  \item for every $x \lexlt y \in \V_\phi$, we have $v^x_\star <_\pi v^y_\star$,
  \item for every $x \in \V_\phi$, we have
    $v^x_a <_\pi v^x_b <_\pi v^x_\top < v^x_c <_\pi v^x_d$ and
    $v^x_a <_\pi v^x_b <_\pi v^x_\bot < v^x_c <_\pi v^x_d$, and finally
  \item for every $x \in \V_\phi$, we have $v^x_\bot <_\pi v^x_\top$ if and only if $\alpha(x) =
    \mathtt{false}$.
  \end{itemize}
  We now show how to construct the threshold graph~$H_\phi^\pi$ from the constructed
  graph~$G_\phi$ by editing exactly~$k_\phi = |\C| \cdot (3 |\V_\phi| - 1)$
  edges.  For a clause~$c$, let~$x$ be any variable satisfying~$c$.  If~$x$
  appears positively, add every non-existing edge from~$v_c$ to every vertex~$v
  \leq_\pi v^x_\top$ and delete all the rest.  If~$x$ appears negated,
  use~$v^x_\bot$ instead.  We break the remainder of the proof in the forward
  direction into two claims:
  
  \begin{claim}
    \label{claim:h-is-threshold}
    $H^\pi_\phi$ is a threshold graph.
  \end{claim}
   
  \begin{proof}[Proof of Claim~\ref{claim:h-is-threshold}]
    Let $G_\phi$ and $\pi$ be given, both adhering to the above construction.  Since
    $G_\phi$ was a split graph, $\pi$ a total ordering of the elements in the
    independent set part and every vertex of the clique part of $H^\pi_\phi$
    sees a prefix of the vertices of the independent set, their neighborhoods
    are naturally nested.  Hence $H^\pi_\phi$ is a threshold graph by
    Proposition~\ref{prop:threshold-nested-neighborhoods}.
  \end{proof}
  
   \begin{claim}
    \label{claim:k-edits}
    $\left| E(G_\phi) \triangle E(H^\pi_\phi) \right| = k_\phi$.
  \end{claim}
 
  \begin{proof}[Proof of Claim~\ref{claim:k-edits}]
    Since we did not edit any of the edges within the clique part nor the
    independent set part, we only need to count the number of edits going
    between a clause vertex and the variable vertices.  Let $c$ be any clause
    and $x$ the lexicographically smallest variable satisfying $c$.  Suppose
    furthermore, without loss of generality, that $x$ appears positively in $c$
    and has thus $\alpha(x) = \mathtt{true}$.  We now show that $|F(v_c)| =
    3|\V_\phi|-1$, and since $c$ was arbitrary, this concludes the proof of the
    claim.  Since $v_c$ is adjacent to exactly three vertices per variable, and
    non-adjacent to exactly three vertices per variable, we added all the edges
    to the vertices appearing before $x$ and removed all the edges to the
    vertices appearing after $x$.  This cost exactly $3 (|\V_\phi|-1) =
    3|\V_\phi|-3$, hence we have two edges left in our budget for $c$.
    Moreover, the edge $v_cv^x_a$ was added and the edge $v_cv^x_d$ was deleted.
    Now, $c$ is adjacent to every vertex to the before, and including, $x$, and
    non-adjacent to all the vertices after $x$.  The budget used was $3
    (|\V_\phi|-1) + 2 = 3|\V_\phi|-1$.  Hence, the total number of edges edited
    to obtain $H^\pi_\phi$ is $\sum_{c \in \mathcal{C}} 3|\V_\phi|-1 = |\C|
    \cdot (3 |\V_\phi| - 1) = k_\phi$.
  \end{proof}
  
  This shows that if $\phi$ is satisfiable, then $(G_\phi,k_\phi)$ is a yes-instance of
  \pname{Threshold Editing}.

  \bigskip
  
  In the reverse direction, let~$(G_\phi, k_\phi)$ be a constructed instance from a
  given 3-CNF-SAT formula~$\phi$ and let~$F$ be a minimal editing set such
  that~$G_\phi \triangle F$ is a threshold graph and~$|F| \leq k_\phi$.  We aim
  to construct a satisfying assignment $\alpha \colon \V_\phi \to
  \{\mathtt{true},\mathtt{false}\}$ from~$G_\phi \triangle F$.  By
  Observation~\ref{obs:g-is-split}, $H = G_\phi \triangle F$ has the same split
  partition as~$G_\phi$.  By construction, we have enforced
  the ordering, $\order$, of each of the vertices corresponding to the
  variables.  Thus, we know exactly how~$H$ looks, with the exception of the
  internal ordering of each literal and its negation.  Construct the
  assignment~$\alpha$ as described above, i.e., $\alpha(x) = \mathtt{false}$ if
  and only if~$v^x_\bot <_\pi v^x_\top$.
  
  By Lemmata~\ref{lem:budget-lower} and~\ref{lem:save-satisfy}, it follows
  directly that~$\alpha$ is a satisfying assignment for~$\phi$ which concludes
  the proof of the main lemma.
\end{proof}

\bigskip

The above lemma shows that there is a polynomial time many-one (Karp) reduction
from \pname{3Sat} to \pname{Threshold Editing} so we may wrap up the main
theorem of this section.
Lemma~\ref{lem:reduc-correct} implies Theorem~\ref{thm:te-np-complete}, that
\TE{} is \NP-complete, even on split graphs.

For the sake of the next section, devoted to the proof of
Theorem~\ref{thm:ce-np-complete}, we define the following annotated version of
editing to threshold graphs.  In this problem, we are given a split graph and we
are asked to edit the graph to a threshold graph while respecting the split
partition.

\defproblem{Split Threshold Editing}
{A split graph $G=(V,E)$ with split partition $(C,I)$, and an integer $k$.}
{Is there an editing set $F \subseteq C \times I$ of size at most $k$ such that $G \triangle
  F$ is a threshold graph?}

\begin{corollary}
  \label{cor:split-te-np-hard}
  \STE{} is \NP-complete.
\end{corollary}
\begin{proof}
  \STE{} is clearly in \NP{} and that the problem is \NP-complete follows
  immediately from combining Lemma~\ref{lem:reduc-correct} with
  Observation~\ref{obs:g-is-split}.
\end{proof}

\begin{corollary}
  \label{cor:te-ste-eth-hard}
  Assuming ETH, neither \TE{} nor \STE{} are solvable in $2^{o(\sqrt k)} \cdot
  \poly(n)$ time.
\end{corollary}

% subsection
%
%
% CHAIN AND CHORDAL STUFF
%
%
\subsection{NP-hardness of Chain Editing and Chordal Editing}
\label{sec:np-hardness-chain-and-chordal}

%
%
% CHAIN EDITING
%
%
\subsubsection{Chain Graphs: Proof of Theorem~\ref{thm:ce-np-complete}}
\label{sec:np-hardness-chain}

A bipartite graph~$G = (A,B,E)$ is a \emph{chain graph} if the neighborhoods
of~$A$ are nested (which necessarily implies the neighborhoods of~$B$ are nested
as well).  Recalling Proposition~\ref{prop:chain-defs}, chain graphs are closely
related to threshold graphs; Given a bipartite graph~$G = (A,B,E)$, if one
replaces~$A$ (or~$B$) by a clique, the resulting graph is a threshold graph if
and only if~$G$ was a chain graph.

It immediately follows from the above exposition that the following problem is
\NP-complete.  This problem has also been referred to as \pname{Chain Editing}
in the literature (for instance in the work by Guo~\cite{guo2007problem}).

\defproblem{Bipartite Chain Editing}
{A bipartite graph $G = (A,B,E)$ and an integer $k$}
{Does there exist a set $F \subseteq A \times B$ of size at most~$k$ such that $G \triangle
  F$ is a chain graph?}

Observe that we in this problem are given a bipartite graph together with a
bipartition, and we are asked to respect the bipartition in the editing set.

\begin{corollary}
  The problem \pname{Bipartite Chain Editing} is \NP-complete.
\end{corollary}
\begin{proof}
  We reduce from \STE{}.  Recall that to this problem, we are given a split
  graph~$G=(V,E)$ with split partition~$(C,I)$, and an integer~$k$, and asked
  whether there
  is an editing set $F \subseteq C \times I$ of size at most~$k$ such that $G \triangle F$ is
  a threshold graph.
  Since a chain graph is a threshold graph with the edges in the clique
  partition removed (Proposition~\ref{prop:chain-defs}), it follows that~$G
  \triangle F$ with all the edges in the clique partition removed is a chain
  graph.
  
  Let~$(G,k)$ be the input to \STE{} and let~$(C,I)$ be the split partition.
  Remove all the edges in~$C$ to obtain a bipartite graph $G' = (A,B,E')$.  Now
  it follows directly from Proposition~\ref{prop:chain-defs} that~$(G,k)$ is a
  yes instance to \STE{} if and only if~$(G',k)$ is a yes instance to \BCE.
\end{proof}

\defproblem{Chain Editing}
{A graph $G = (V,E)$ and a non-negative integer $k$}
{Is there a set $F$ of size at most~$k$ such that $G \triangle F$ is a chain
  graph?}

We now aim to prove Theorem~\ref{thm:ce-np-complete}, that \pname{Chain Editing}
is \cclass{NP}-complete.

\begin{proof}[Proof of Theorem~\ref{thm:ce-np-complete}]
  Reduction from \BCE.  Let $G = (A,B,E)$ be a bipartite graph and consider the
  input instance~$(G,k)$ to \BCE.  We now show that adding~$2(k+1)$ new edges
  to~$G$ to obtain a graph~$G' = (V,E')$, gives us that~$(G',k)$ is a yes
  instance for \pname{Chain Editing} if and only if~$(G,k)$ is a yes instance
  for \BCE.
  
  Let $G = (A,B,E)$ be a bipartite graph and~$k$ a positive integer.  Add~$k+1$
  new vertices $a_1, \cdots a_{k+1}$ to~$A$ and make them universal to~$B$, and
  add~$k+1$ new vertices $b_1, \cdots b_{k+1}$ to~$B$ and make them universal
  to~$A$.  Call the resulting graph~$G' = (V,E')$.
  
  The following claim follows immediately from the construction.
  \begin{claim}
    If $G' \triangle F$ is a chain graph with $|F| \leq k$, then $G' \triangle F$
    has bipartition $(A \cup \{a_1, \dots, a_{k+1}\}, B \cup \{b_1, \dots, b_{k+1}\})$.
  \end{claim}

  It follows that for any input instance~$(G,k)$ to \BCE, the instance~$(G',k)$
  as constructed above is a yes instance for \CE{} if and only if~$(G,k)$ is a
  yes instance for \BCE.
\end{proof}

\begin{corollary}
  Assuming ETH, there is no algorithm solving neither \CE{} nor \BCE{} in time
  $2^{o(\sqrt k)} \cdot \poly(n)$.
\end{corollary}
\begin{proof}
  In both these cases we reduced from \STE{} without changing the parameter~$k$.
  Hence this follows immediately from the above exposition and from
  Corollary~\ref{cor:te-ste-eth-hard}.
\end{proof}

%
%
%  CHORDAL EDITING
%
%
\subsubsection{Chordal Graphs}
We will now combine our previous result on \pname{Chain Editing} with the
following observation of Yannakakis to prove Theorem~\ref{thm:chordal-edit-npc}.
Yannakakis showed~\cite{yannakakis1981computing}, while proving the
\NP-completeness of \pname{Chordal Completion} (more often known as
\pname{Minimum Fill-In}~\cite{fomin2012subexponential}), that a bipartite graph
can be transformed into a chain graph by adding at most~$k$ edges if and only if
the cobipartite graph formed by completing the two sides can be transformed into
a chordal graph by adding at most~$k$ edges.

\begin{theorem}
  \label{thm:chordal-edit-npc}
  \pname{Chordal Editing} is \NP-hard.
\end{theorem}

To prove the theorem, we will first give an intermediate problem that makes the
proof simpler.  Let $G=(A,B,E)$ be a cobipartite graph.  Define the problem
\pname{Cobipartite Chordal Editing} to be the problem which on input $(G,k)$
asks if we can edit at most~$k$ edges between~$A$ and~$B$, i.e., does there
exist an editing set~$F \subseteq A \times B$ of size at most~$k$, such that~$G
\triangle F$ is a chordal graph.  That is, \pname{Cobipartite Chordal Editing}
asks for the bipartition~$A,B$ to be respected.

\defproblem{Cobipartite Chordal Editing}
{A cobipartite graph $G = (A,B,E)$ and an integer $k$}
{Does there exist a set $F \subseteq A \times B$ of size at most $k$ such that $G \triangle
  F$ is a chordal graph?}

We will use the following observation to prove the above theorem:
\begin{lemma}
  \label{lem:chor-eq-cobip}
  If $G=(A,B,E)$ is a bipartite graph, and $G' = (A,B,E')$ is the cobipartite
  graph constructed from~$G$ by completing~$A$ and~$B$, then~$F$ is an optimal
  edge editing set for \pname{Bipartite Chain Editing} on input $(G,k)$ if and
  only if~$F$ is an optimal edge editing set for \pname{Cobipartite Chordal
    Editing} on input $(G',k)$.
\end{lemma}
\begin{proof}
  Let $F$ be an optimal editing set for \BCE{} on input $(G,k)$ and suppose that
  $G' \triangle F$ has an induced cycle of length at least four.  Since~$G'$ is
  cobipartite, it has a cycle of length exactly four.  Let $a_1b_1b_2a_2a_1$ be
  this cycle.  But then it is clear that $a_1b_1, a_2b_2$ forms an
  induced~$2K_2$ in~$G \triangle F$, contradicting the assumption that~$F$ was
  an editing set.

  For the reverse direction, suppose~$F$ is an optimal edge editing set for
  \CCE{} on input $(G',k)$ only editing edges between~$A$ and~$B$.  Suppose for
  the sake of a contradiction that~$G \triangle F$ was not a chain graph.
  Since~$F$ only goes between $A$ and $B$, $G \triangle F$ is bipartite and
  hence by the assumption must have an induced~$2K_2$.  This obstruction must be
  on the form $a_1b_1,a_2b_2$, but then $a_1b_1b_2a_2a_1$ is an induced~$C_4$
  in~$G' \triangle F$ which is a contradiction to the assumption that~$G'
  \triangle F$ was chordal.  Hence~$G \triangle F$ is a chain graph.
\end{proof}

\begin{corollary}
  \label{cor:cce-np}
  \CCE{} is \NP-complete.
\end{corollary}

We are now ready to prove Theorem~\ref{thm:chordal-edit-npc}.
\begin{proof}[Proof of Theorem~\ref{thm:chordal-edit-npc}]
  Let $(G = (A,B,E),k)$ be a cobipartite graph as input to \CCE.  Our reduction
  is as follows.  Create $G' = (A' \cup B',E')$ as follows:
  \begin{itemize}
  \item $A' = A \cup \{a_1, a_2, \dots, a_{k+1}\}$,
  \item $B' = B \cup \{b_1, b_2, \dots, b_{k+1}\}$,
  \item $E' = E \cup \bigcup_{i \leq k+1, b \in B'} \{a_ib\} \cup \bigcup_{i,j \leq k+1}
    \{a_ia_j, b_ib_j\}$
  \end{itemize}
  Finally, we create $G''$ as follows.  For every edge $a_ia_j$ create $k+1$ new
  vertices adjacent to only $a_i$ and $a_j$.  Do the same thing for every edge
  $b_ib_j$.  This forces none of the edges in $A'$ to be removed and none of the
  edges in $B'$ to be removed.
  \begin{claim}
    The instance of \pname{Chordal Editing} $(G'',k)$ is equivalent to the
    instance $(G,k)$ to \CCE{}.
  \end{claim}
  \begin{proof}[Proof of claim]
    The proof of the above claim is straight-forward.  If we delete an edge
    within $A$ (resp.~$B$), we create at least $k+1$ cycles of length $4$, each
    of which uses at least one edge to delete, hence in any yes instance, we do
    not edit edges within $A$ (resp.~$B$).  Furthermore, any chordal graph
    remains chordal when adding a simplicial vertex, which is exactly what the
    $k+1$ new vertices are.
  \end{proof}
  
  From the claim it follows that $(G'',k)$ is a yes instance to \pname{Chordal
    Editing} if and only if $(G,k)$ is a yes instance to \CCE.  The theorem
  follows immediately from Corollary~\ref{cor:cce-np}.
\end{proof}

\begin{corollary}
  Assuming ETH, there is no algorithm solving \pname{Chordal Editing} in time
  $2^{o(\sqrt k)} \cdot \poly(n)$.
\end{corollary}

% Section 4
\newcommand{\regular}{regular}     % good
\newcommand{\outlying}{outlying}   % bad
\newcommand{\important}{important} % ugly

\section{Kernels for Modifications into Threshold and Chain Graphs}
\label{sec:kernel}

First we give kernels with quadratically many vertices for the following three
problems: \pname{Threshold Completion}, \pname{Threshold Deletion}, and
\pname{Threshold Editing}, answering a recent question of Liu, Wang and
Guo~\cite{liu2014overview}.  Then we continue by providing kernels with
quadratically many vertices for \pname{Chain Completion}, \pname{Chain
  Deletion}, and \pname{Chain Editing}.
Our kernelization algorithms uses techniques similar to the previous result that
\pname{Trivially Perfect Editing} admits a polynomial
kernel~\cite{drange2014apolynomial}.
Observe that the class of threshold graphs is closed under taking complements.
It follows that for every instance~$(G,k)$ of \pname{Threshold
  Completion},~$(\bar{G}, k)$ is an equivalent instance of \pname{Threshold
  Deletion} (and vice versa).  Almost the same trick applies to \pname{Chain
  Deletion}.  Due to this, we restrict our attention to the completion and
editing variants for the remainder of the section.  Motivated by the
characterization of threshold graphs in
Propositions~\ref{prop:c4p42k2-free}~and~\ref{prop:chain-defs}, we define
obstructions (also see Figure~\ref{fig:forbidden-graphs}).

\begin{definition}[$\mathcal{H}$, Obstruction]
  \label{def:obstruction}
  A graph~$H$ is a \emph{threshold obstruction} if it is isomorphic to a member
  of the set~$\{ C_4 , P_4 , 2K_2\}$ and a \emph{chain obstruction} if it is
  isomorphic to a member of the set~$\{C_3, 2K_2, C_5\}$.  If it is clear from
  the context, we will often use the term obstruction for both threshold and
  chain obstructions and denote the set of obstructions by~$\mathcal{H}$.
  Furthermore, if an obstruction~$H$ is an induced subgraph of a graph~$G$ we
  call~$H$ an \emph{obstruction in~$G$}.
\end{definition}

\begin{definition}[Realizing]
  \label{def:realize}
  For a graph~$G$ and a set of vertices $X \subseteq V(G)$ we say that a vertex $v \in
  V(G) \setminus X$ is \emph{realizing} $Y \subseteq X$ if $N_X(v) = Y$.
  Furthermore, we say that a set $Y \subseteq X$ is \emph{being realized} if
  there is a vertex~$v \in V(G) \setminus X$ such that~$v$ is realizing~$Y$.
\end{definition}

Before proceeding, we observe that our kernelization algorithms does not modify
any edges, and only changes the budget in the case that we discover that we have
a no-instance (in which case we return~$(H,0)$, where~$H$ is an obstruction
in~$G$).  The only modification of the instance is to delete vertices, hence the
kernelized instance is an induced subgraph of the original graph.  Since the
parameter is never increased, we obtain \emph{proper kernels}.

\subsection{Modifications into Threshold Graphs}
\label{sec:modif-into-thresh}

We now focus on modifications to threshold graphs and obtaining kernels for
these operations.

\subsubsection{Outline of the Kernelization Algorithm}
\label{sec:outl-kern-algor}

The kernelization algorithm consists of a twin reduction rule and an irrelevant
vertex rule.  The twin reduction rule is based on the %following two
observation % s: First, if an obstruction contains a vertex~$v$ one can replace~$v$
% by a twin not in the obstruction and obtain a new, isomorphic obstruction.
% And second, no solution can interact with all the vertices in a large twin
% class, for some meaning of large.  This implies
that any obstruction containing vertices from a large enough twin class will
have to be handled by edges not incident to the twin class.  From this
observation, we may conclude that for any twin class, we may keep only a certain
amount
% would still be large if we removed a vertex from it, then we can do so
without affecting the solutions.

A key concept of the irrelevant vertex rule is what will be referred to as a
\emph{\modulator{}}.  A \modulator{} is a set of vertices~$X$ in~$G$ of linear
size in~$k$, such that for every obstruction~$H$ in~$G$ one can add and remove
edges in~$[X]^2$ to turn~$H$ into a non-obstruction.  First, we prove that we
can in polynomial time either obtain such a set~$X$ or conclude correctly that
the instance is a no-instance.  The observation that~$G-X$ is a threshold graph
will be exploited heavily and we now fix a threshold decomposition
$(\mathcal{C}, \mathcal{I})$ of $G-X$.  We then prove that the idea of
Proposition~\ref{prop:threshold-nested-neighborhoods} can be extended to
vertices in~$G-X$ when considering their neighborhoods in~$G$.  In other words,
the neighborhoods of the vertices in~$G-X$ are nested also when considering~$G$.
This immediately yields that the number of subsets of~$X$ that are being
realized is bounded linearly in the size of~$X$ and hence also in~$k$.

We now either conclude that the graph is small or we identify a sequence of
levels in the threshold decomposition containing many vertices, such that all
the clique vertices and all the independent set vertices in the sequence have
identical neighborhoods in~$X$, respectively.  The crux is that in the middle of
such a sequence there will be a vertex that is replaceable by other vertices in
every obstruction and hence is irrelevant.  Such a sequence is obtained by
discarding all levels in the decomposition that are extremal with respect to a
subset~$Y$ of~$X$, meaning that there either are no levels above or underneath
that contain vertices realizing~$Y$.  One can prove that in this process, only a
quadratic number of vertices are discarded and from this we obtain a kernel.

\subsubsection{The Twin Reduction Rule}
\label{sec:twin-reduction-rule}

First, we introduce the twin reduction rule as described above.  For the
remainder of the section we will assume this rule to be applied exhaustively and
hence we can assume all twin classes to be small.

\begin{redrule}[Twin reduction rule]
  \label{redrule:twin-reduction}
  Let~$(G,k)$ be an instance of \TC{} or \TE{} and~$v$ a vertex in~$G$ such
  that~$|\tc(v)| > 2k+2$.  We then reduce the instance to~$(G-v, k)$.
\end{redrule}

\begin{lemma}
  \label{lemma:twin-reduction-valid}
  Let~$G$ be a graph and~$v$ a vertex in~$G$ such that $|\tc(v)| > 2k+2$.  Then
  for every~$k$ we have that $(G,k)$ is a yes-instance of \TC{} (or \TE) if and
  only if $(G-v,k)$ is a yes-instance of \TC{} (resp.\ \TE).
\end{lemma}

\begin{proof}
  For readability we only consider \TC{}, however the exact same proof works for
  \TE{}.  Let~$G' = G-v$.  It trivially holds that if~$(G,k)$ is a yes-instance,
  then also~$(G',k)$ is a yes-instance.  This is due to the fact that removing a
  vertex never will create new obstructions.
  
  Now, let~$(G',k)$ be a yes-instance and assume for a contradiction
  that~$(G,k)$ is a no-instance.  Let~$F$ be an optimal solution of~$(G',k)$
  and~$W$ an obstruction in~$(G \triangle F,k)$.  Since~$W$ is not an
  obstruction in~$G'$ it follows immediately that~$v$ is in~$W$.  Furthermore,
  since~$|F| \leq k$ it follows that there are two vertices $a,b \in \tc(v)
  \setminus \left\{ v \right\}$ that~$F$ is not incident to.  Also, one can
  observe that no obstruction contains more than two vertices from a twin class
  and hence we can assume without loss of generality that~$b$ is not in~$W$.  It
  follows that $N_{G \triangle F}(v) \cap (W-v) = N_G(v) \cap (W-v) = N_G(b)
  \cap (W-v) = N_{G'}(b) \cap (W-v)$ and hence the graph induced on $V(W)
  \triangle \{ b,v \}$ is an obstruction in $G' \triangle F$, contradicting
  that~$F$ is a solution.
\end{proof}

\subsubsection{The Modulator}
\label{sec:modulator}

To obtain an~$O(k^2)$ kernel we aim at an irrelevant vertex rule.  However, this
requires some tools.  The first one is the concept of a \modulator{}, as defined
below.

\begin{definition}[Threshold modulator]
  \label{def:modulator}
  Let~$G$ be a graph and~$X \subseteq V(G)$ a set of vertices.  We say that~$X$ is a
  \emph{\modulator{}} of~$G$ if for every obstruction~$W$ in~$G$ it holds that
  there is a set of edges~$F$ in~$[X]^2$ such that~$W \triangle F$ is not an
  obstruction.
\end{definition}

Less formally, a set~$X$ is a \modulator{} of a graph~$G$ if for every
obstruction~$W$ in~$G$ you can edit edges between vertices in~$X$ to turn~$W$
into a non-obstruction.  Our kernelization algorithm will heavily depend on
finding a small \modulator{}~$X$ and the fact that~$G-X$ is a threshold graph.

\begin{lemma}
  \label{lemma:modulator-construction}
  There is a polynomial time algorithm that given a graph~$G$ and an integer~$k$
  either
  \begin{itemize}
  \item outputs a \modulator{}~$X$ of~$G$ such that~$|X| \leq 4k$ or
  \item correctly concludes that~$(G,k)$ is a no-instance of both
    \pname{Threshold Completion} and \pname{Threshold Editing}.
  \end{itemize}
\end{lemma}

\begin{proof}
  Let~$X_1$ be the empty set and $\mathcal{W} = \{W_1, \dots, W_t\}$ the set of all
  obstructions in~$G$.  We execute the following procedure for every $W_i$ in
  $\mathcal{W}$: If $W_i \triangle F$ is an obstruction for every $F \subseteq
  [X_i \cap V(W_i)]^2$ we let $X_{i+1} = X_i \cup V(W_i)$, otherwise we let
  $X_{i+1} = X_i$.  After we have considered all obstructions we let $X =
  X_{t+1}$.  If $|X| > 4k$ we conclude that $(G,k)$ is a no-instance, otherwise
  we output~$X$.
  
  Since all obstructions are finite the algorithm described clearly runs in
  polynomial time.  We now argue that~$X$ is a \modulator{} of~$G$.  If~$W_i$
  was added to~$X_{i+1}$, we let~$F$ be all the non-edges of~$W$.  Since~$W
  \triangle F$ is isomorphic to~$K_4$ it follows immediately that $W \triangle
  F$ is not an obstruction.  If~$W_i$ was not added to~$X_{i+1}$, let~$F$ the
  set found in $[X_i \cap V(W_i)]^2$ such that $W_i \triangle F$ is not an
  obstruction.  Observe that~$F \subseteq [X]^2$ and hence~$X$ is a \modulator{}.
  
  It remains to prove that if $|X| > 4k$ then $(G,k)$ is a no-instance of
  \pname{Threshold Editing}.  Observe that it will follow immediately
  that~$(G,k)$ is a no-instance of \pname{Threshold Completion}.  Since every
  obstruction consists of four vertices there was at least~$k+1$ obstructions
  added during the procedure.  Assume without loss of generality that~$W_1,
  \dots, W_{k+1}$ was added.  Observe that by construction, a solution must
  contain an edge in~$[X_{i+1}-X_{i}]^2$ for every~$i \in [k+1]$ and hence
  contains at least~$k+1$ edges.
\end{proof}

%% Begin figure: C4, P4, 2K2
\begin{figure}[t]
  
  \centering
  \begin{tikzpicture}
    \tikzset{scale=0.8}
    \tikzset{Vertex/.style={shape=circle,draw,scale=0.9}}
    \tikzset{Edge/.style={}}
    \tikzset{Dots/.style={scale=0.8}}

    \foreach \e/\x/\y [count=\k] in {
      f00/-1/0,
      f01/-1/1,
      f02/-1/2,
      f03/-1/3
    }
    \node[Vertex, fill=black!0] (\e) at (\x,\y) {};

    \foreach \e/\x/\y [count=\k] in {
      f10/0.5/0,
      f11/0.5/1,
      f12/0/2,
      f13/1/2
    }
    \node[Vertex, fill=black!13] (\e) at (\x,\y) {};
 
    \foreach \e/\x/\y [count=\k] in {
      f20/2/0,
      f21/2/1,
      f22/2/2,
      f23/3/1
    }
    \node[Vertex, fill=black!26] (\e) at (\x,\y) {};  

    \foreach \e/\x/\y [count=\k] in {
      f30/4/0,
      f31/5/0,
      f32/4/1,
      f33/5/1
    }
    \node[Vertex, fill=black!39] (\e) at (\x,\y) {};

    \foreach \e/\x/\y [count=\k] in {
      f40/6/0,
      f41/7/0,
      f42/6/1,
      f43/7/1
    }
    \node[Vertex, fill=black!52] (\e) at (\x,\y) {};
    
    \foreach \e/\x/\y [count=\k] in {
      f50/8/0,
      f51/9/0,
      f52/8/1,
      f53/9/1
    }
    \node[Vertex, fill=black!65] (\e) at (\x,\y) {};

    \foreach \e/\x/\y [count=\k] in {
      f60/10/0,
      f61/11/0,
      f62/10/1,
      f63/11/1
    }
    \node[Vertex, fill=black!78] (\e) at (\x,\y) {};

%    \foreach \e/\x/\y [count=\k] in {
%      f00/-1/0,
%      f01/-1/1,
%      f02/-1/2,
%      f03/-1/3,
%      f10/0.5/0,
%      f11/0.5/1,
%      f12/0/2,
%      f13/1/2,
%      f20/2/0,
%      f21/2/1,
%      f22/2/2,
%      f23/3/1,
%      f30/4/0,
%      f31/5/0,
%      f32/4/1,
%      f33/5/1,
%      f40/6/0,
%      f41/7/0,
%      f42/6/1,
%      f43/7/1,
%      f50/8/0,
%      f51/9/0,
%      f52/8/1,
%      f53/9/1,
%      f60/10/0,
%      f61/11/0,
%      f62/10/1,
%      f63/11/1
%    }
%    \node[Vertex] (\e) at (\x,\y) {};

    \node[] () at (-0.95,-1) {$H_1$};
    \node[] () at (0.55,-1) {$H_2$};
    \node[] () at (2.05,-1) {$H_3$};
    \node[] () at (4.55,-1) {$H_4$};
    \node[] () at (6.55,-1) {$H_5$};
    \node[] () at (8.55,-1) {$H_6$};
    \node[] () at (10.55,-1) {$H_7$};
    
    \foreach \a/\b in 
    {{f00/f01},{f01/f02},{f02/f03},{f10/f11},{f12/f13},
      {f20/f21},{f20/f23},{f21/f22},{f30/f32},{f32/f33},{f31/f33},
      {f40/f41},{f41/f43},{f43/f42},{f42/f40},{f50/f52},{f51/f53},
      {f60/f62},{f60/f61},{f61/f63}}
    \draw[Edge](\a) to node {} (\b);
    
    \draw [dotted, rounded corners] (-1.5,-0.5) rectangle (11.5,0.5);
    \node[scale=1.5] () at (12,0) {$X$};
  \end{tikzpicture}
  
  \caption{Some of the intersections of an obstruction with a \modulator{} $X$
  that will not occur by definition. More specifically the ones necessary for
  the proof of the kernel.}
  \label{fig:modulator}
\end{figure}
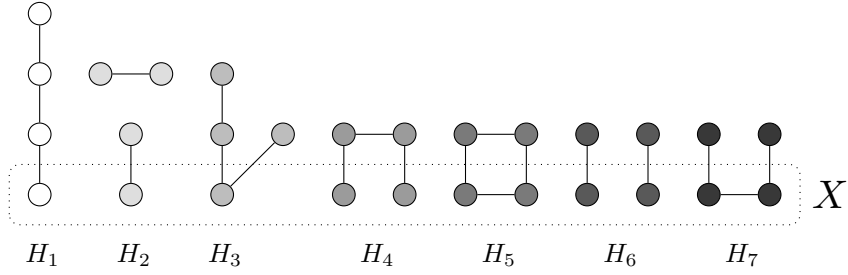

\subsubsection{Obtaining Structure}
\label{sec:obtaining-structure}

We now exploit the \modulator{} and its interaction with the remaining graph to
obtain structure.  First, we prove that the neighborhoods of the vertices
outside of~$X$ are nested and that the number of realized sets in~$X$ are
bounded linearly in~$k$.

\begin{lemma}
  \label{lemma:G-X-ordering}
  Let~$G$ be a graph and~$X$ a \modulator{}.  For every pair of vertices~$u$
  and~$v$ in~$G-X$ it holds that either $N(u) \subset N[v]$ or $N(v) \subset N[u]$.
\end{lemma}
\begin{proof}
  Assume otherwise for a contradiction and let $u'$ be a vertex in $N(u)
  \setminus N[v]$ and $v'$ a vertex in $N(v) \setminus N[u]$.  Let $W =
  G[\left\{ u,v,u',v' \right\}]$ and observe that $uu'$ and $vv'$ are edges
  in~$W$ and $uv'$ and $vu'$ are non-edges in~$W$ by definition.  Hence, no
  matter if some of the edges $uv$ and $u'v'$ are present or not,~$W$ is an
  obstruction in~$G$ (see Figure~\ref{fig:modulator} for an illustration).
  Since $u'v'$ is the only pair in~$W$ possibly with both elements in~$X$ this
  contradicts~$X$ being a \modulator{}.
\end{proof}

\begin{lemma}
  \label{lemma:bounded-neighborhoods}
  Let~$G$ be a graph and~$X$ a corresponding \modulator{}, then
  \[
  \left| \left\{ N_X(v) \text{ for } v \in V(G) \setminus X \right \} \right| \leq |X| + 1 .
  \]
  Or in other words, there are at most $|X|+1$ sets of~$X$ that are being
  realized.
\end{lemma}
\begin{proof}
  Let~$u$ and~$v$ be two vertices in $G-X$.  It follows directly from
  Lemma~\ref{lemma:G-X-ordering} that either $N_X(v) \subseteq N_X(u)$ or
  $N_X(v) \supseteq N_X(u)$.  The result follows immediately.
\end{proof}

With the definition of the modulator and the basic properties above, we are now
ready to extract more vertices from the instance, aiming at many consecutive
levels that have the same neighborhood in~$X$ for the clique, and independent
set vertices, respectively.  This will lead up to our irrelevant vertex rule.

Let~$G$ be a graph,~$X$ a \modulator{} and $(\mathcal{C},\mathcal{I})$ a
threshold partition of $G-X$.  Letting~$P$ denote either~$C$ or~$I$, we say that
a subset $Y \subseteq X$ has its \emph{upper extreme} in~$P_i$ if~$P_i$
realizes~$Y$ and for every $j > i$ it holds that~$P_j$ does not realize~$Y$.
Similarly, a subset $Y \subseteq X$ has its \emph{lower extreme} in~$P_i$
if~$P_i$ realizes~$Y$ and for every $j < i$ it holds that~$P_j$ does not
realize~$Y$.  We say that $Y \subseteq X$ is \emph{extremal} in~$P_i$ if~$Y$ has
its upper or lower extreme in~$Y$.  Observe that every $Y \subseteq X$ is
extremal in at most two clique fragments and two independent set fragments.

We continue having $P$ denote either $C$ or $I$.
\begin{lemma}
  \label{lemma:sequential-realization}
  Let~$G$ be a graph,~$X$ a \modulator{} and $(\mathcal{C},\mathcal{I})$ a
  threshold partition of $G-X$.  For every $Y \subseteq X$ it holds that if~$Y$
  has its lower extreme in $P_{\ell}$ and upper extreme in $P_{u}$, then for
  every vertex $v \in P_i$ with $i \in [\ell+1, u-1]$ it holds that $N_X(v) =
  Y$.
\end{lemma}
\begin{proof}
  Let~$Y$ be a subset of~$X$ with $C_{\ell}$ and $C_{u}$ being its lower and
  upper extremes in the clique respectively.  By definition there is a vertex $u
  \in C_{\ell}$ and a vertex $w \in C_u$ such that $N_X(u) = N_X(w) = Y$.
  Let~$i$ be an integer in $[\ell+1, u-1]$ and a vertex $v \in C_i$.  By the
  definition of a threshold partition it holds that $N_{G-X}(w) \subset
  N_{G-X}(v) \subset N_{G-X}(u)$.  It follows from
  Lemma~\ref{lemma:G-X-ordering} that $N(w) \subset N[v]$ and that $N(v) \subset
  N[u]$.  Hence, \[Y = N_X(w) \subseteq N_X(v) \subseteq N_X(u) = Y\] and we
  conclude that $N_X(v) = Y$.  Since~$i$ and~$v$ was arbitrary, the proof is
  complete.
\end{proof}

\begin{definition}[Important, Outlying, and Regular]
  We say that $P_i$ in the partition is \emph{\important{}} if there is a $Y
  \subseteq X$ such that~$Y$ has its extreme in $P_i$.  Furthermore, a level
  $L_i$ is \important{} if $C_i$ or $I_i$ is \important{}.  Let~$f$ be the
  smallest number such that $|\cup_{i \leq f} C_i| \geq 2k+2$ and~$r$ the
  largest number such that $|\cup_{i \geq r} I_i| \geq 2k+2$.  A level $L_i$ is
  \emph{\outlying{}} if $i \leq f$ or $i \geq r$.  All other levels of the
  decomposition are \emph{\regular{}} and a vertex is \regular{}, \outlying{} or
  \important{} depending on the type of the level it is contained in.
\end{definition}

\begin{lemma}
  \label{obs:bounded-number-of-ugly-levels}
  Let~$G$ be a graph and~$X$ a \modulator{} of~$G$ of size at most~$4k$.  Then
  every threshold partition of~$G-X$ has at most~$16k+4$ \important{} levels.
\end{lemma}
\begin{proof}
  The result follows immediately from the definition of important levels and
  Lemma~\ref{lemma:bounded-neighborhoods}.
\end{proof}

\begin{lemma}
  \label{lemma:bounded-number-of-ugly-levels}
  Let~$G$ be a graph,~$X$ a \modulator{} of~$G$ and $(\mathcal{C},\mathcal{I})$
  a threshold partition of~$G-X$, then for every set~$Y \subseteq X$ there are
  at most two important clique fragments (independent fragments) realizing~$Y$.
\end{lemma}
\begin{proof}
  We first prove the statement for clique fragments.  Let~$Y$ be a subset of~$X$
  and $i < j < k$ three integers.  Assume for a contradiction that $C_i, C_j$
  and $C_k$ are important clique fragments all realizing~$Y$.  By definition
  there are vertices $u \in C_i$, $v \in C_j$ and $w \in C_k$ such that $N_X(u)
  = N_X(v) = N_X(w) = Y$.  Furthermore, there is a vertex $v' \in C_j$ such that
  $N_X(v') \neq Y$ since $C_j$ is important and~$Y$ does not have an extreme in
  $C_j$.  By the definition of threshold partitions, we have that $N_{G-X}(w)
  \subset N_{G-X}(v') \subset N_{G-X}(u)$.  Lemma~\ref{lemma:G-X-ordering}
  immediately implies that $N(w) \subset N[v']$ and $N(v') \subset N[u]$ and
  since $\left\{ u,v',w \right\} \subseteq \cup \mathcal{C}$ it holds that $N[u]
  \subseteq N[v'] \subseteq N[w]$.  Since $N_X(v') \neq Y$, we have $N_X(w)
  \subset N_X(v') \subset N_X(u)$, which contradicts the definition of~$w$
  and~$u$ since $N_X(u) = N_X(w)$.  By a symmetric argument, the statement also
  holds for independent fragments.
\end{proof}

\begin{lemma}
  \label{lemma:bounding-uglyness}
  Let~$G$ be a graph,~$X$ a \modulator{} of~$G$ of size at most $4k$ and
  $(\mathcal{C},\mathcal{I})$ a threshold partition of $G-X$.  Then there are at
  most $64k^2+80k+16$ \important{} vertices in $G-X$.
\end{lemma}
\begin{proof}
  Let~$Y$ be the set of all vertices contained in a \important{} clique or
  independent fragment and let~$Z$ be the set of all \important{} vertices.
  Observe that $Y \subseteq Z$ and that every $C_i$ or $I_i$ contained in $Z
  \setminus Y$ is a twin class in~$G$ by definition.  By
  Lemma~\ref{obs:bounded-number-of-ugly-levels} there are at most $16k+4$
  \important{} levels and since the twin-rule has been applied exhaustively it
  holds that $|Z \setminus Y| \leq (16k+4)(2k+2) = 32k^2+40k+8$.
  
  Let~$A$ be a subset of~$X$ and~$B$ the vertices in~$Y$ such that their
  neighborhood in~$X$ is exactly~$A$.  Let~$D$ be a $C_i$ or $I_i$ contained
  in~$Y$ and observe that $D \cap B$ is a twin class in~$G$ and hence $|D \cap
  B| \leq 2k+2$.  And hence it follows from
  Lemma~\ref{lemma:bounded-number-of-ugly-levels} that $|B| \leq 8k+8$.
  Furthermore, we know from Lemma~\ref{lemma:bounded-neighborhoods} that there
  are at most $4k+1$ realized in~$X$ and hence $|Y| \leq (8k+8)(4k+1) =
  32k^2+40k+8$.  It follows immediately that $|Z| \leq 64k^2 + 80k + 16$,
  completing the proof.
\end{proof}

\begin{lemma}
  \label{lemma:bounding-badness}
  Let~$G$ be a graph,~$X$ a \modulator{} of~$G$ of size at most $4k$ and
  $(\mathcal{C},\mathcal{I})$ a threshold partition of $G-X$.  Then there are at
  most $80k^2 + 112k + 32$ \important{} and \outlying{} vertices in total in
  $G-X$.
\end{lemma}
\begin{proof}
  By Lemma~\ref{lemma:bounding-uglyness} it follows that there are at most
  $64k^2 + 80k + 16$ vertices that are \important{} and possibly \outlying{}.
  It follows from Lemma~\ref{lemma:sequential-realization} that if a level is
  not \important{} its vertices are covered by at most two twin classes in~$G$
  and hence the level contains at most $4k+4$ vertices.  By definition there are
  at most $4k+4$ \outlying{} levels and hence at most $(4k+4)(4k+4) = 16k^2 +
  32k + 16$ vertices which are \outlying{}, but not \important{}.  The result
  follows immediately.
\end{proof}

\begin{lemma}
  \label{lemma:good-stays-around}
  Let~$G$ be a graph,~$X$ a \modulator{} of~$G$,~$v$ a \regular{} vertex in some
  threshold partition $(\mathcal{C}, \mathcal{I})$ of $G-X$, $C = \cup
  \mathcal{C}$ and $I = \cup \mathcal{I}$.  Then for every $F \subseteq
  [V(G)]^2$ such that $G \triangle F$ is a threshold graph, $|F| \leq k$ and
  every split partition $(C_F, I_F)$ of $G \triangle F$ we have:
  \begin{itemize}
  \item $v \in C$ if and only if $v \in C_F$ and
  \item $v \in I$ if and only if $v \in I_F$.
  \end{itemize}
\end{lemma}
\begin{proof}
  Observe that the two statements are equivalent and that it is sufficient to
  prove the forward direction of both statements.  First, we prove that $v \in
  C$ implies that $v \in C_F$.  Let~$Y$ be the set of \outlying{} vertices in $I
  \cap N_G(v)$ and recall that $|Y| > 2k+1$ by definition.  Observe that at most
  $2k$ vertices in~$Y$ are incident to~$F$ and hence there are two vertices $u,
  u'$ in~$Y$ that are untouched by~$F$.  Clearly,~$u$ and $u'$ are not adjacent
  in $G \triangle F$ and hence we can assume without loss of generality that~$u$
  is in $I_F$.  Since~$u$ is untouched by~$F$,~$v$ is adjacent to~$u$ by the
  definition of \outlying{} vertices and hence~$v$ is not in $I_F$.  A symmetric
  argument gives that $v \in I$ implies that $v \in I_F$ and hence our argument
  is complete.
\end{proof}

%
% IRRELEVANT VERTEX RULE
%

\subsubsection{The Irrelevant Vertex Rule}
\label{sec:irrel-vert-rule}

We have now obtained the structure necessary to give our irrelevant vertex rule.
But before stating the rule, we need to define these consecutive levels with
similar neighborhood and what it means for a vertex to be in the middle of such
a collection of levels.

\begin{definition}[Large strips, central vertices]
  \label{def:strip}
  Let~$G$ be a graph,~$X$ a \modulator{} and $(\mathcal{C},\mathcal{I})$ a
  threshold partition of $G-X$.  A \emph{strip} is a maximal set of consecutive
  levels which are all \regular{} and we say that a strip is \emph{large} if it
  contains at least $16k+13$ vertices.  For a strip $S =
  ([C_a,I_a],\dots,[C_b,I_b])$ a vertex $v \in C_i$ is \emph{central} if $a \leq
  i \leq b$ and $|\cup_{j \in [a,i-1]} C_j| \geq 2k+2$ and $|\cup_{j \in
    [i+1,b]} C_j| \geq 2k+2$.  Similarly we say that a vertex $v \in I_i$ is
  \emph{central} if $a \leq i \leq b$ and $|\cup_{j \in [a,i-1]} I_j| \geq 2k+2$
  and $|\cup_{j \in [i+1,b]} I_j| \geq 2k+2$.  Furthermore, we say that a
  vertex~$v$ is \emph{central in~$G$} if there exists a \modulator{}~$X$ of size
  at most $4k$ and a threshold decomposition of $G-X$ such that~$v$ is central
  in a large strip.
\end{definition}

\begin{lemma}
  \label{lemma:large-strip-have-central-vertex}
  If a strip is large it has a central vertex.
\end{lemma}
\begin{proof}
  Let $S = ([C_a, I_a],\dots,[C_b,I_b])$ be a large strip.  First, we consider the
  case when $|\cup_{i \in [a,b]} C_i| \geq |\cup_{i \in [a,b]} I_i|$.  Observe
  that $|\cup_{i \in [a,b]} C_i| \geq 8k+7$.  Let~$i$ be the smallest number
  such that $|\cup_{j \in [a,i-1]} C_j| \geq 2k+2$.  It follows immediately from
  $|C_{i-1}| \leq 2k+2$ that $|\cup_{j \in [a,i-1]} C_j| \leq 4k+3$.
  Furthermore, since $|C_i| \leq 2k+2$ it follows that $|\cup_{j \in [i+1,b]}
  C_j| \geq 8k+7-(2k+2+4k+3) = 2k+2$.  And hence any vertex in $C_i$ is central.
  A symmetric argument for the case $|\cup_{i \in [a,b]} C_i| < |\cup_{i \in
    [a,b]} I_i|$ completes the proof.
\end{proof}

\begin{redrule}[Irrelevant vertex rule]
  \label{redrule:irrelevant-vertex-reduction}
  If $(G,k)$ be an instance of \TC{} or \TE{} and~$v$ is a central vertex
  in~$G$, reduce to $(G-v, k)$.
\end{redrule}

\begin{lemma}
  \label{lemma:irrelevant-vertex-rule}
  Let $(G,k)$ be an instance,~$X$ a \modulator{} and~$v$ a central vertex
  in~$G$.  Then $(G,k)$ is a yes-instance of \TE{} (\TC{}) if and only if $(G-v,
  k)$ is a yes-instance.
\end{lemma}
\begin{proof}
  For readability we only consider \TE{}, however the exact same proof works for
  \TC{}.  For the forwards direction, for any vertex~$v$, if $(G,k)$ is a
  yes-instance, then $(G-v,k)$ is also a yes-instance.  This holds since
  threshold graphs are hereditary.
  
  For the reverse direction,
  let $(G-v,k)$ be a yes-instance and assume for a contradiction that $(G,k)$ is
  a no-instance.
  Let~$F$ be a solution of $(G-v,k)$ satisfying
  Lemma~\ref{lemma:no-switching-of-nestedness}, and let $G' = G \triangle F$.
  By assumption, $(G,k)$ is a no-instance, so specifically, $G'$ is not a
  threshold graph.  Let~$W$ be an obstruction in~$G'$.
  Clearly $v \in W$ since otherwise there is an obstruction in $(G-v) \triangle
  F$, so consider $Z = V(W) - v$.  For convenience we will use $N'$ to denote
  neighborhoods in $G'$ and specifically for any set $Y \subseteq V(G')$,
  $N'_Y(v) = N_{G'}(v) \cap Y$.
  Furthermore, let $(\mathcal{C}, \mathcal{I})$ be a threshold decomposition of
  $G-X$ such that there is a large strip~$S$ for which~$v$ is central.  We will
  now consider the case when~$v$ is in the clique of $G-X$.  Since $|F| \leq k$
  and~$S$ is a large strip it follows immediately that there are two clique
  vertices~$w$ and $w'$ in~$S$ in higher levels than~$v$ that is not incident
  to~$F$.  Observe that $\left\{ w, w', v \right\}$ forms a triangle and
  that~$W$ contains no such subgraph.  Hence, we can assume without loss of
  generality that $w \notin V(W)$.  Similarly, we obtain a clique vertex~$u$ in
  a lower level than~$v$ in~$S$ such that $u \notin W$.
  
  \smallskip
  
  Observe that $G'[Z \cup \left\{ u \right\}]$ is not an obstruction and hence
  $N_Z(u) = N'_Z(u) \neq N'_Z(v) = N_Z(v)$.  Since~$u$ and~$v$ are clique
  vertices from the same strip it is true that $N_X(v) = N_X(u)$ and hence there
  is an independent vertex~$a$ in~$Z$ such that $\level(u) \leq \level(a) <
  \level(v)$ (see Definition~\ref{def:threshold-partition}).
  In other words~$u$ is adjacent to~$a$ while~$v$ and~$w$ are not.  By a
  symmetric argument we obtain a vertex~$b$ such that $\level(v) \leq \level(b)
  < \level(w)$, meaning that both~$u$ and~$v$ are adjacent to~$b$ while~$w$ is
  not.  Let~$y$ be last vertex of~$Z$, meaning that $\left\{ v,y,a,b \right\} =
  V(W)$.  Observe that~$a$ and~$b$ are \regular{} vertices and hence it follows
  from Lemma~\ref{lemma:good-stays-around} that for every threshold partition of
  $G'$ it holds that $\left\{ a,b \right\}$ are independent vertices.
  
  % [t]
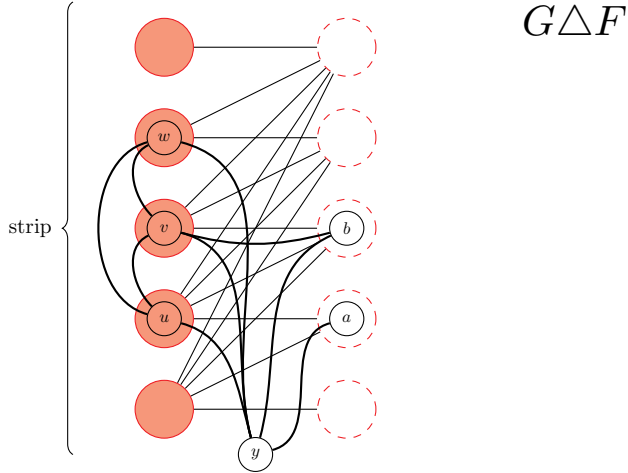
\begin{figure}[t]
  \centering
  \begin{tikzpicture}[every node/.style={circle,draw,scale=.8},scale=.6]
    \node[lab,scale=2] (GX) at (9,8.5) {$G \triangle F$};
    
    \node[scale=1.3, cliq] (c1) at (0, 8) {};
    \node[scale=1.3, cliq] (c2) at (0, 6) {};
    \node[scale=1.3, cliq] (c3) at (0, 4) {};
    \node[scale=1.3, cliq] (c4) at (0, 2) {};
    \node[scale=1.3, cliq] (c5) at (0, 0) {};
    
    \node[scale=1.3,indep] (i1) at (4, 8) {};
    \node[scale=1.3,indep] (i2) at (4, 6) {};
    \node[scale=1.3,indep] (i3) at (4, 4) {};
    \node[scale=1.3,indep] (i4) at (4, 2) {};
    \node[scale=1.3,indep] (i5) at (4, 0) {};
    
    \draw[decorate,decoration={brace,amplitude=4pt}] (-2,-1) --
    (-2,9) node [midway,fill=none,draw=none,xshift=-2em] {strip};

    %% between c and i
    \foreach \k in {1,...,5} {
      \foreach \j in {\k,...,5} {
        \draw (c\j) -- (i\k);
      }
    }
    
    \node[minimum size=2em,inner sep=0,scale=.8] (v) at (0, 4) {$v$};
    
    \node[minimum size=2em,inner sep=0,scale=.8] (va) at (0, 6) {$w$};
    
    \draw[thick] (v) to[out=135, in=205] (va);
    
    \node[minimum size=2em,inner sep=0,scale=.8] (vb) at (0, 2) {$u$};
    
    \draw[thick] (v)  to[out=205, in=135] (vb);
    \draw[thick] (va) to[out=195, in=165] (vb);

    \node[minimum size=2em,inner sep=0,scale=.8] (u) at (4, 4) {$b$};

    \draw[thick] (v) to[out=-15,in=195] (u);

    \node[minimum size=2em,inner sep=0,scale=.8] (ub) at (4, 2) {$a$};

    \node[minimum size=2em,inner sep=0,scale=.8] (x) at (2, -1) {$y$};
    
    \draw[thick] (ub) to[out=195, in=15] (x);
    \draw[thick] (va) to[out=-15, in=105] (x);
    \draw[thick] (vb) to[out=-15, in=105] (x);

    \draw[thick] (v) to[out=-15, in=105] (x);
    
    \draw[thick] (x) to[out=75, in=205] (u);
  \end{tikzpicture}
  \caption{The vertex $v$ was a center vertex in a strip and $W = \{v,a,b,y\}$ was
    assumed to be an obstruction.}
  \label{fig:irrelevant-vertex-rule}
\end{figure}

  \smallskip
  
  Recall that $u,v,w,a,b$ are all regular and hence they are in the same
  partitions in $G'$ as in $G-X$ by Lemma~\ref{lemma:good-stays-around}.
  Furthermore, since~$W$ is an obstruction and~$a$ is neither adjacent to~$v$
  nor~$b$ in $G'$ it holds that~$y$ and~$a$ are adjacent in $G'$.  It follows
  that~$y$ is a clique vertex in $G'$ and hence it is adjacent to both~$u$
  and~$w$ in $G'$.  Since~$u$ and~$w$ are not incident to~$F$ by definition,
  they are adjacent to~$y$ also in~$G$.  Since $u,v,w$ are regular and from the
  same strip it follows that~$v$ is adjacent to~$y$ in both~$G$ and $G'$.
  Observe that the only possible adjacency not yet decided in~$W$ is the one
  between~$b$ and~$y$.  However, for~$W$ to be an obstruction it should not be
  present.  Hence~$y$ is adjacent to~$a$ but not to~$b$ in $G'$.  By definition
  $N_G(a) \subseteq N_G(b)$, however by the last observation this is not true in
  $G'$.  This contradicts that~$F$ satisfies
  Lemma~\ref{lemma:no-switching-of-nestedness}.  A symmetric argument gives a
  contradiction for the case when~$v$ is an independent vertex and hence the
  proof is complete.
\end{proof}

\bigskip

The above lemma shows the soundness of the irrelevant vertex rule,
Rule~\ref{redrule:irrelevant-vertex-reduction}, and we may therefor apply it
exhaustively.  The following theorem wraps up the goal of this section.

\begin{theorem}
  \label{theorem:threshold-kernel}
  The following three problems admit kernels with at most $336k^2 + 388k + 92$
  vertices: \TD{}, \TC{} and \TE{}.
\end{theorem}
\begin{proof}
  Assume that Rules~\ref{redrule:twin-reduction}
  and~\ref{redrule:irrelevant-vertex-reduction} have been applied exhaustively.
  If this process does not produce a \modulator{}, we can safely output a
  trivial no-instance by Lemma~\ref{lemma:modulator-construction}.  Hence, we
  can assume that we have a \modulator{}~$X$ of size at most $4k$ and that the
  reduction rules cannot be applied.  By Lemma~\ref{lemma:bounding-badness} we
  know that there are at most $80k^2+112k+32$ vertices in $G-X$ that are not
  \regular{}.  Furthermore, every \regular{} vertex is contained in a strip and
  by Lemma~\ref{obs:bounded-number-of-ugly-levels} there are at most $16k+5$
  such strips.  Since the reduction rules cannot be applied, no strip is large,
  and hence they contain at most $16k+12$ vertices each.  Since every vertex
  in~$G$ is either in~$X$, or considered \regular{}, \outlying{} or \important{}
  this gives us $4k + 80k^2 + 112k + 32 + (16k+5)(16k+12) = 336k^2 + 388k + 92$
  vertices in total.
\end{proof}

% subsection
\subsection{Adapting the Kernel to Modification to Chain Graphs}
\label{sec:adapt-kern-modif}

In this section we provide kernels with quadratically many vertices for
\pname{Chain Deletion}, \pname{Chain Completion} and \pname{Chain Editing}.  Due
to the fundamental similarities between modification to chain and threshold
graphs we omit the full proof and instead highlight the differences between the
two proofs.  Observe that the only proofs for the threshold kernels that
explicitly applies the obstructions are those of
Lemmata~\ref{lemma:modulator-construction},~\ref{lemma:G-X-ordering}~and~\ref{lemma:irrelevant-vertex-rule}
and hence these will receive most of our attention.

The twin reduction rule goes through immediately and hence our first obstacle is
the modulator.  Luckily, this is a minor one.  Recall from
Definition~\ref{def:obstruction} that the obstructions now are $\mathcal{H} =
\{2K_2, C_3, C_5\}$; We thus get a \cmodulator{}~$X$ of size~$5k$, as the
largest obstruction contains five vertices.  Besides this detail, the proof goes
through exactly as it is.

%% Begin figure: chain-modulator obstruction
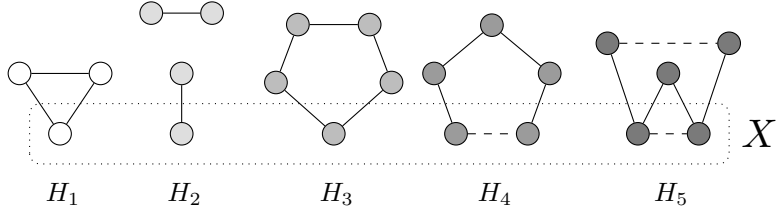
\begin{figure}[htp]
  \centering
  \begin{tikzpicture}
    \tikzset{scale=0.8}
    \tikzset{Vertex/.style={shape=circle,draw,scale=0.9}}
    \tikzset{Edge/.style={}}
    \tikzset{Dots/.style={scale=0.8}}

    \foreach \e/\x/\y [count=\k] in {
      f00/-1/0,
      f01/-1.666/1,
      f02/-0.333/1
    }
    \node[Vertex, fill=black!0] (\e) at (\x,\y) {};

    \foreach \e/\x/\y [count=\k] in {
      f10/1/0,
      f11/1/1,
      f12/0.5/2,
      f13/1.5/2
    }
    \node[Vertex, fill=black!13] (\e) at (\x,\y) {};
 
    \foreach \e/\x/\y [count=\k] in {
      f20/3.5/0,
      f21/2.55/0.85,
      f22/4.45/0.85,
      f23/2.91/1.81,
      f24/4.09/1.81
    }
    \node[Vertex, fill=black!26] (\e) at (\x,\y) {};  

    \foreach \e/\x/\y [count=\k] in {
      f30/6.1/1.8,
      f31/5.15/1.,
      f32/7.05/1.,
      f33/5.51/0,
      f34/6.69/0
    }
    \node[Vertex, fill=black!39] (\e) at (\x,\y) {};

    \foreach \e/\x/\y [count=\k] in {
      f40/8/1.5,
      f41/8.5/0,
      f42/9/1,
      f43/9.5/0,
      f44/10/1.5
    }
    \node[Vertex, fill=black!52] (\e) at (\x,\y) {};

%
%    \foreach \e/\x/\y [count=\k] in {
%      f30/4/0,
%      f31/5/0,
%      f32/4/1,
%      f33/5/1
%    }
%    \node[Vertex, fill=black!39] (\e) at (\x,\y) {};
%
%    \foreach \e/\x/\y [count=\k] in {
%      f40/6/0,
%      f41/7/0,
%      f42/6/1,
%      f43/7/1
%    }
%    \node[Vertex, fill=black!52] (\e) at (\x,\y) {};
%    
%    \foreach \e/\x/\y [count=\k] in {
%      f50/8/0,
%      f51/9/0,
%      f52/8/1,
%      f53/9/1
%    }
%    \node[Vertex, fill=black!65] (\e) at (\x,\y) {};
%
%    \foreach \e/\x/\y [count=\k] in {
%      f60/10/0,
%      f61/11/0,
%      f62/10/1,
%      f63/11/1
%    }
%    \node[Vertex, fill=black!78] (\e) at (\x,\y) {};

%    \foreach \e/\x/\y [count=\k] in {
%      f00/-1/0,
%      f01/-1/1,
%      f02/-1/2,
%      f03/-1/3,
%      f10/0.5/0,
%      f11/0.5/1,
%      f12/0/2,
%      f13/1/2,
%      f20/2/0,
%      f21/2/1,
%      f22/2/2,
%      f23/3/1,
%      f30/4/0,
%      f31/5/0,
%      f32/4/1,
%      f33/5/1,
%      f40/6/0,
%      f41/7/0,
%      f42/6/1,
%      f43/7/1,
%      f50/8/0,
%      f51/9/0,
%      f52/8/1,
%      f53/9/1,
%      f60/10/0,
%      f61/11/0,
%      f62/10/1,
%      f63/11/1
%    }
%    \node[Vertex] (\e) at (\x,\y) {};

    \node[] () at (-0.95,-1) {$H_1$};
    \node[] () at (1.05,-1) {$H_2$};
    \node[] () at (3.55,-1) {$H_3$};
    \node[] () at (6.15,-1) {$H_4$};
    \node[] () at (9.05,-1) {$H_5$};
%    \node[] () at (8.55,-1) {$H_6$};
%    \node[] () at (10.55,-1) {$H_7$};
    
    \foreach \a/\b in 
    {{f00/f01},{f01/f02},{f02/f00},{f10/f11},{f12/f13},
    {f20/f21},{f21/f23},{f23/f24},{f24/f22},{f22/f20},
    {f30/f31},{f31/f33},{f34/f32},{f32/f30},
    {f40/f41},{f41/f42},{f42/f43},{f43/f44}}
    \draw[Edge](\a) to node {} (\b);

    \foreach \a/\b in 
    {{f33/f34}, {f40/f44}, {f41/f43}}
    \draw[Edge, dashed](\a) to node {} (\b);

    \draw [dotted, rounded corners] (-1.5,-0.5) rectangle (10,0.5);
    \node[scale=1.5] () at (10.5,0) {$X$};
  \end{tikzpicture}
  \caption{Some of the intersections of an obstruction with a \cmodulator{} $X$
    that by definition will not occur. Dashed edges represent edges that could
    or could not be there. These are the intersections necessary for the proof
    of the kernel.}
  \label{fig:chain-modulator}
\end{figure}

%%% Local Variables:
%%% mode: latex
%%% TeX-master: "../paper"
%%% End:

\subsubsection{An Additional Step}
\label{sec:an-additional-step}

Before we continue with the remainder of the proof we need an additional step.
Namely to discard all vertices that are isolated in $G-X$.  We will prove that
by doing this we discard at most $O(k^2)$ vertices.  Now, if the irrelevant
vertex rule concludes that the graph is small, then the graph is small also when
we reintroduce the discarded vertices.  And if we find an irrelevant vertex, we
remove it and reintroduce the discarded vertices before we once again apply our
reduction rules.  Due to the locality of our arguments, this is a valid
approach.

\begin{lemma}
    \label{lemma:bounded-independent-vertices}
    For a graph~$G$ and a corresponding \cmodulator{}~$X$ there are at most
    $10k^2 + 12k + 2$ isolated vertices in $G-X$.
\end{lemma}

\begin{proof}
  Let~$I$ be the set of isolated vertices in $G-X$.  We will prove that
  $\mathcal{F} = \{N_X(v) \mid v \in I \}$ is laminar (see
  Definition~\ref{def:laminar}) and hence by Lemma~\ref{lem:laminar-bounded} it
  holds that $|\mathcal{F}| \leq |X|+1 \leq 5k+1$.  It follows immediately, due
  to the twin reduction rule, that there are at most $(5k+1)(2k+2) = 10k^2 + 12k
  + 2$ independent vertices in $G-X$.

  Assume for a contradiction that there are vertices $u, v$ and~$w$ in~$I$ such
  that there exists $u' \in N_X(u) \setminus N_X(v)$ and $v' \in N_X(v)
  \setminus N_X(u)$ with $\{u', v'\} \subseteq N_X(w)$.  These vertices
  intersect with the modulator as a variant of the forbidden~$H_5$ in
  Figure~\ref{fig:chain-modulator} and hence we get a contradiction.
\end{proof}

\subsubsection{Nested Neighborhoods}
\label{sec:nested-neighborhoods}

From now on we will assume in all of our arguments that there are no isolated
vertices in $G-X$.  The next difference is with respect to
Lemma~\ref{lemma:G-X-ordering}, which is just not true anymore.  The lemma
provided us with the nested structure of the neighborhoods in the modulator and
was crucial for most of the proofs.  As harmful as this appears to be at first,
it turns out that we can prove a weaker version that is sufficient for our
needs.

\begin{lemma}[New, weaker version of Lemma~\ref{lemma:G-X-ordering}]
  \label{lemma:chain-G-X-ordering}
  Let~$G$ be a graph and~$X$ a \cmodulator{}.  For every pair of vertices~$u$
  and~$v$ in the \emph{same bipartition} of $G-X$ it holds that either $N(u)
  \subseteq N(v)$ or $N(v) \subseteq N(u)$.
\end{lemma}
\begin{proof}
  Let~$u$ and~$v$ be two vertices from the same bipartition of $G-X$.  By the
  definition of chain graphs we can assume that $N_{G-X}(u) \subseteq
  N_{G-X}(v)$.  Assume for a contradiction that the lemma is not true.  Then
  there is a vertex $u' \in N_X(u) \setminus N_X(v)$ and a vertex~$v'$ in
  $N_X(v) \setminus N_X(u)$.  By definition,~$u$ and~$v$ are not adjacent.
  Since there are no isolated vertices in $G-X$ there is a vertex $a \in
  N_{G-X}(u) \subseteq N_{G-X}(v)$.  Observe that if~$a$ is adjacent to
  either~$u'$ or~$v'$ we get a~$C_3$ that only has one vertex in~$X$, which is a
  contradiction (see~$H_1$ in Figure~\ref{fig:chain-modulator}).  However,
  if~$a$ is not adjacent to both~$u'$ and~$v'$ then $\{u, v, u', v', a\}$ forms
  the same interaction with the modulator as~$H_4$ in
  Figure~\ref{fig:chain-modulator} and hence our proof is complete.
\end{proof}

One can observe that Lemma~\ref{lemma:chain-G-X-ordering} is a sufficiently
strong replacement for Lemma~\ref{lemma:G-X-ordering} since all proofs are
applying the lemma to vertices from only one partition of $G-X$.  The only
exception is the proof of Lemma~\ref{lemma:bounded-neighborhoods}, but by
applying Lemma~\ref{lemma:chain-G-X-ordering} on one partition at the time we
obtain the following bound instead:
\[
\left|\left\{N_X(v) \text{ for } v \in V(G) \setminus X \right\}\right| \leq 2|X|+2 .
\]

%
% IRRELEVANT VERTEX RULE
%

\subsubsection{An Irrelevant Vertex Rule}
\label{sec:an-irrelevant-vertex}

It only remains to prove that the irrelevant vertex rule can still be applied
with this new set of obstructions.  Although the strategy is the same, the
details are different and hence we provide the proof in full detail.

\begin{lemma}
  \label{lemma:chain-irrelevant-vertex-rule}
  Let $(G,k)$ be an instance,~$X$ a \modulator{} and~$v$ a central vertex
  in~$G$.  Then $(G,k)$ is a yes-instance of \pname{Chain Editing} (\pname{Chain
    Completion}) if and only if $(G-v, k)$ is a yes-instance.
\end{lemma}
\begin{proof}
  For readability we only consider \pname{Chain Editing}, however the exact same
  proof works for \pname{Chain Completion}.  For the forwards direction, for any
  vertex~$v$, if $(G,k)$ is a yes-instance, then $(G-v,k)$ is also a
  yes-instance.  This holds since chain graphs are hereditary.
  
  For the reverse direction,
  let $(G-v,k)$ be a yes-instance and assume for a contradiction that $(G,k)$ is
  a no-instance.
  Let~$F$ be a solution of $(G-v,k)$ satisfying
  Lemma~\ref{lemma:no-switching-of-nestedness}, and let $G' = G \triangle F$.
  By assumption, $(G,k)$ is a no-instance, so specifically,~$G'$ is not a chain
  graph.  Let~$W$ be an obstruction in~$G'$.
  Clearly $v \in W$, since otherwise there is an obstruction in $(G-v) \triangle
  F$.  Let $Z = V(W) - v$.  For convenience we will use~$N'$ to denote
  neighborhoods in~$G'$ and specifically for any set $Y \subseteq V(G')$,
  $N'_Y(v) = N_{G'}(v) \cap Y$.
  Furthermore, let $(\mathcal{A}, \mathcal{B})$ be a chain decomposition of
  $G-X$ such that there is a large strip~$S$ for which~$v$ is central.  Let $A =
  \cup \mathcal{A}$ and $B = \cup \mathcal{B}$.  We will now consider the case
  when~$v$ is in~$A$.  Since $|F| \leq k$ and~$S$ is a large strip it follows
  immediately that there are two vertices~$w$ and~$w'$ in $A \cap S$ in higher
  levels than~$v$ that is not incident to~$F$.  Observe that $\left\{ w, w', v
  \right\}$ forms an independent set of size three and that~$W$ contains no such
  subgraph.  Hence, we can assume without loss of generality that $w \notin
  V(W)$.  Similarly, we obtain a vertex~$u$ in~$A$ at a lower level than~$v$
  in~$S$ such that $u \notin W$.
  
  \smallskip
  
  Observe that $G'[Z \cup \left\{ u \right\}]$ is not an obstruction and hence
  $N_Z(u) = N'_Z(u) \neq N'_Z(v) = N_Z(v)$.  Since~$u$ and~$v$ are vertices
  in~$A$ from the same strip it is true that $N_X(v) = N_X(u)$ and hence there
  is a vertex~$a$ in $Z \cap B$ such that $\level(u) \leq \level(a) <
  \level(v)$.  In other words~$u$ is adjacent to~$a$, while~$v$ and~$w$ are not.
  By a symmetric argument we obtain a vertex~$b$ such that $\level(v) \leq
  \level(b) < \level(w)$, meaning that both~$u$ and~$v$ are adjacent to~$b$
  while~$w$ is not.  We now fix a chain decomposition $(\mathcal{A}',
  \mathcal{B}')$ and let $A' = \cup \mathcal{A}'$ and $B' = \cup \mathcal{B}'$.
  Observe that~$a$ and~$b$ are \regular{} vertices and hence it follows from the
  chain version of Lemma~\ref{lemma:good-stays-around} that $\left\{ a,b
  \right\}$ is in~$B'$.  This yields immediately that~$W$ is not a~$C_3$
  (since~$a$ and~$b$ are not adjacent) and hence we are left the cases of~$W$
  being a~$2K_2$ or a~$C_5$.
  
  We now consider the case when~$W$ is isomorphic to a $2K_2$.  Let~$y$ be the
  last vertex of~$Z$, meaning that $\left\{ v,y,a,b \right\} = V(W)$.  Observe
  that since~$W$ is a $2K_2$ it holds that~$y$ is adjacent to~$a$, but not
  to~$b$.  However, in~$G$ it holds that $N(a) \subseteq N(b)$ and hence~$F$ is
  not satisfying Lemma~\ref{lemma:no-switching-of-nestedness}, which is a
  contradiction.
  
  Hence we are left with the case that~$W$ is isomorphic to a~$C_5$.  Let $y, x$
  be the last vertices of~$Z$.  Observe that all vertices in~$W$ should be of
  degree two and hence~$a$ is adjacent to both~$x$ and~$y$.  Recall that~$a$ is
  in~$B'$ and observe that~$u$ is in~$A'$ by the same reasoning.  Due to their
  adjacency to~$a$, also~$x$ and~$y$ is in~$A'$.  It follows immediately that
  $u, x$ and~$y$ form an independent set in $(G-v) \triangle F$.  Since~$u$
  and~$v$ are not touched by~$F$ and in the same strip it follows that $v, x$
  and~$y$ form an independent set in~$G'$.  We observe that by this~$W$ can not
  be isomorphic to a~$C_5$.  The argument for the case when $v \in B$ is
  symmetrical and hence the proof is complete.
\end{proof}

We immediately obtain our kernelization results for modifications into chain
graphs by the same wrap up as for threshold graphs.

\begin{theorem}
  \label{theorem:chain-kernel}
  The following three problems admit kernels with at most $O(k^2)$ vertices:
  \pname{Chain Deletion}, \pname{Chain Completion} and \pname{Chain Editing}.
\end{theorem}

% Section 5
\section{Subexponential Time Algorithms}
\label{sec:subept}

\subsection{Threshold Editing in Subexponential Time}
\label{sec:subept-threshold}

In this section we give a subexponential time algorithm for \pname{Threshold
  Editing}.  We also show that we can modify the algorithm to work with
\pname{Chain Editing}.  Combined with the results of Fomin and
Villanger~\cite{fomin2012subexponential} and Drange et al.~\cite{DrangeFPV14},
we now have complete information on the subexponentiality of edge modification
to threshold and chain graphs.  In this section we aim to prove the following
theorem:

\begin{theorem}
  \label{thm:t-e-subept}
  \pname{Threshold Editing} admits a $2^{O(\sqrt k \log k)} + \poly(n)$
  subexponential time algorithm.
\end{theorem}
The additive $\poly(n)$ factor comes from the kernelization procedure of
Section~\ref{sec:kernel}.  The remainder of the algorithm operates on the
kernel, and thus has running time that only depends on~$k$.

We will throughout refer to a \emph{solution}~$F$.  In this case, we are
assuming a given input instance~$(G,k)$, and then~$F$ is a set of at most~$k$
edges such that~$G \triangle F$ is a threshold graph.  In the next section,
Section~\ref{sec:subept-chain}, we will assume~$G \triangle F$ to be a chain
graph.  Furthermore, after Section~\ref{sec:getting-partition}, we will be
working with the problem \STE{}, so we assume~$F \subseteq C \times I$
when~$(C,I)$ is the split partition of~$G$.

\begin{definition}[Potential split partition]
  Given a graph~$G$ and an integer~$k$ (called the budget), for~$C$ and~$I$ a
  partitioning of~$V(G)$ we call~$(C,I)$
  a \emph{potential split partition} of~$G$ provided that
  \[
  \binom{|C|}{2} - E(C) + E(I) \leq k .
  \]
  That is, the cost of making~$G$ into a split graph with the prescribed
  partitioning does not exceed the budget.
\end{definition}

\paragraph{A brief explanation of the algorithm for
  Theorem~\ref{thm:t-e-subept}.}
% 
% Step 1 / 4
% 
The algorithm consists of four parts, the first of which is the kernelization
algorithm described in Section~\ref{sec:kernel}.  This gives in polynomial time
an equivalent instance~$(G,k)$ with the guarantee that $|V(G)| = O(k^2)$.  We
may observe that this is a \emph{proper kernel}, i.e., the reduced instance's
parameter is bounded by the original parameter.  This allows us to use time
subexponential in the kernelized parameter.

\smallskip

% Step 2 / 4
The second step in the algorithm selects a potential split partitioning of~$G$.
We show that the number of such partitionings is bounded subexponentially
in~$k$, and that we can enumerate them all in subexponential time.
This step actually also immediately implies that editing%
\footnote{Indeed, editing to split graphs is solvable in linear
  time~\cite{hammer1981splittance}.},
completing and deleting to split graphs can be solved in subexponential time,
however all of this was known~\cite{hammer1981splittance,ghosh2013faster}.  The
main part of this step is Lemma~\ref{lemma:list-relevant-split-partitions}.
For the remainder of the algorithm, we may thus assume that the input instance
is a split graph, and that the split partition needs to be preserved, that is,
we focus on solving \STE.

\smallskip

% Step 3 / 4 & Step 4 / 4
The third and fourth steps of the algorithm consists of repeatedly finding
special kind of separators and solving structured parts individually;
Step three consists of locating so-called \emph{cheap vertices} (see
Definition~\ref{def:cheap} for a formal explanation).  These are vertices,~$v$,
whose neighborhood is almost correct, in the sense that there is an optimal
solution in which~$v$ is incident to only~$O(\sqrt k)$ edges.  The dichotomy of
cheap and expensive vertices gives us some tools for decomposing the graph.
Specific configurations of cheap vertices allow us to extract three parts, one
part is a highly structured part, the second part is a provably small part which
me may brute force, and the last part we solve recursively.  All of which is
done in subexponential time $2^{O(\sqrt k \log k)}$.

\bigskip

Henceforth we will have in mind a ``target graph''~$H = G \triangle F$ with
threshold partitioning~$(\C, \I)$.  We refer to the set of edges~$F$ as the
\emph{solution}, and assume $|F| \leq k$.

\subsubsection{Getting the Partition}
\label{sec:getting-partition}

As explained above, a crucial part of the algorithm is to enumerate all sets of
size at most $O(\sqrt k)$.  The following lemma shows that this is indeed doable
and we will use the result of this lemma throughout this section without
necessarily referring to it.

\begin{lemma}
  \label{lem:enum-small}
  For every $c \in \mathbb{N}$ there is an algorithm that, given an input instance
  $(G,k)$ with $|V(G)| = k^{O(1)}$ enumerates all vertex subsets of size $c
  \sqrt{k}$ in time $2^{O(\sqrt{k} \log k)}$.
\end{lemma}
\begin{proof}
  Given an input graph $G=(V,E)$, with $|V| = n = k^{O(1)}$ and a natural
  number~$k$ we can simply output the family of sets $\mathcal{X} \subseteq 2^{V}$ of
  size at most $c \sqrt k$, which takes time
  \[
  \sum_{\kappa \leq c \sqrt{k}}\binom{n}{\kappa} \leq c \sqrt{k} \binom{n}{c\sqrt{k}} \leq c
  \sqrt{k} \cdot n^{c \sqrt{k}} = 2^{O(\sqrt{k} \log n)} = 2^{O(\sqrt{k} \log
    k)} ,
  \]
  where the first inequality follows since $\binom n i$ is increasing for~$i$
  from~$1$ to $c \sqrt k$.
\end{proof}

\bigskip

The second step of the subexponential time algorithm was as described above to
compute the potential split partitionings of the input instance.  Since we are
given a general graph, we do not know immediately which vertices will go to the
clique partition and which will go to the independent set partition.  However,
we now show that there is at most subexponentially many potential split
partitionings.  That is, there are subexponentially many partitionings of the
vertex set into~$(C,I)$ such that it is possible to edit the input graph to a
threshold graph with the given partitioning not exceeding the prescribed budget.

\bigskip

The next lemma will be crucial in our algorithm, as our algorithm presupposes a
fixed split partition.  Using this result, we may in subexponential time compute
every possible split partition within range, and run our algorithm for
completion to threshold graphs on each of these split graphs.

\begin{lemma}[Few split partitions]
    \label{lemma:list-relevant-split-partitions}
    \label{lem:potential-split}
    There is an algorithm that given a graph~$G$ and an integer~$k$ with $|V(G)|
    = k^{O(1)}$, can generate a set $\mathcal{P}$ of split partitions of $V(G)$
    such that for every split graph~$H$ such that $|E(H) \triangle E(G)| \leq k$
    and every split partition $(C,I)$ of~$H$ it holds that $(C,I)$ is an element
    of $\mathcal{P}$.
    Furthermore, the algorithm terminates in $2^{O(\sqrt k \log k )}$ time.
\end{lemma}
\begin{proof}
  Let~$G=(V,E)$ be a graph and~$k$ a natural number.  The first thing we do is
  to guess the size~$s_c$ of the clique and let~$C$ be a set of~$s_c$ vertices
  of highest degrees, and $s_i = n - s$, and let $I = V(G) \setminus C$.
  In the case that $\min\{s_c,s_i\} \leq 6\sqrt k$ we can simply enumerate every
  partitioning by Lemma~\ref{lem:enum-small}, so we assume from now on that
  $\min\{s_c,s_i\} > 6\sqrt k$.
  
  \begin{claim}
    In any split graph~$H$ with $|E(H) \triangle E(G)|\leq k$, where~$H$ has split
    partition $C',I'$ with $|C'| = s_c$, $|C \triangle C'| \leq 2\sqrt k$ and
    $|I \triangle I'| \leq 2\sqrt k$.
  \end{claim}
  
  \begin{proof}
    Suppose that $2 \sqrt k$ vertices $C'$ move from~$C$ to~$I$ and that $2
    \sqrt k$ vertices $I'$ move from~$I$ to~$C$.  Let $\sigma_c = \sum_{v \in
      C'}\deg(v)$ and $\sigma_i = \sum_{v \in I'}\deg(v)$.  First, since the
    vertices are ordered by degree, $\sigma_i \leq \sigma_c$.  Second, since in
    the final solution, $C'$ is in the independent set, $\sigma_c \leq s_c
    2\sqrt k + k$ (we might delete up to~$k$ vertices from $C'$) and using the
    same reasoning, $\sigma_i \geq (s_c - 2\sqrt k) + \binom{2 \sqrt k}{2} -k =
    s_c 2 \sqrt k - 3k - \sqrt k$ (we might add up to~$k$ vertices to $I'$).
    
    However, since $s_c \geq 6\sqrt k$, we have
    \begin{align*}
      s_c \cdot 2 \sqrt k - 3k - \sqrt k \leq \sigma_i &\leq \sigma_c \leq s_c \cdot 2 \sqrt k +
      k\text{, and thus}\\
      9k - \sqrt k \leq \sigma_i &\leq \sigma_c \leq 13 k,
    \end{align*}
    yielding that $\sigma_c \geq 9k - \sqrt k$.  However, we can only lower the total
    degree of $C'$ by $2k$, which means that even if we spend the entire budget
    on deleting from $C'$, $\sum_{v \in C'}\deg_H(v) \geq 6k$ which means that
    there is a vertex in $C'$ with degree higher than the size of the clique (a
    contradiction).
  \end{proof}
  
  Observe that since~$s_c$ and~$s_i$ are fixed, if we move~$\ell$ vertices
  from~$C$ to~$I$, we have to move~$\ell$ vertices from~$I$ to~$C$.  Hence, if
  the claim holds, we can simply enumerate every set of $4\sqrt k$ vertices and
  take the sets with equally many on each side and swap their partition.  Adding
  each such partition to~$\mathcal{P}$ gives the set in question.
\end{proof}

We would like to remark that this lemma also gives a simpler algorithm for
\pname{Split Completion} (equivalently \pname{Split Deletion}).  Ghosh et
al.~\cite{ghosh2013faster} showed that \pname{Split Completion} can be solved in
time $2^{O(\sqrt k \log k)} \cdot \poly(n)$ using the framework of Alon,
Lokshtanov and Saurabh~\cite{alon2009fast}.  However, the following observation
immediately yields a very simple combinatorial argument for the existence of
such an algorithm.
Together with the polynomial kernel by Guo~\cite{guo2007problem}, the following
result is immediate from the above lemma.
\begin{corollary}
  The problem \pname{Split Completion} %(equivalently \pname{Split Deletion})
  is solvable in time $2^{O(\sqrt k \log k)} + \poly(n)$.
\end{corollary}
\begin{proof}
  The algorithm is as follows.  On input $(G,k)$ we compute, using
  Lemma~\ref{lem:potential-split}, every potential split partitioning~$(C,I)$ at
  most~$k$ edges away from~$G$.  Then we in linear time check that~$I$ is indeed
  independent and that~$C$ lacks at most~$k$ edges from being complete.
\end{proof}

\bigskip

\subsubsection{Cheap or Expensive?}
\label{sec:cheap-or-expensive}

We will from now on assume that all our input graphs $G = (V,E)$ are split
graphs provided with a split partition $(C,I)$, and that we are to solve \STE{},
that is, we have to respect the split partitioning.  We are allowed to do this
with subexponential time overhead, as per the previous section and specifically
Lemma~\ref{lemma:list-relevant-split-partitions}.  In addition, we assume that
$|V(G)| = O(k^2)$.

Given an instance~$(G,k)$ and a solution~$F$, we define the \emph{editing
  number} of a vertex~$v$, denoted~$\en^F_G(v)$, to be the number of edges
in~$F$ incident to a vertex~$v$.
% i.e.,~$\en^F_G(v) = |\{ u \in V(G) \mid uv \in F\}|$.
When~$G$ and~$F$ are clear from the context, we will simply write~$\en(v)$.  A
vertex~$v$ will be referred to as \emph{cheap} if~$\en(v) \leq 2\sqrt{k}$ and
\emph{expensive} otherwise.  We will call a set of vertices~$U \subseteq V$
\emph{small} provided that~$|U| \leq 2\sqrt{k}$ and \emph{large} otherwise.

\begin{definition}
  \label{def:cheap}
  Given an instance $(G,k)$ with solution~$F$, we call a vertex~$v$ \emph{cheap}
  if $\en(v) \leq 2\sqrt k$.
\end{definition}

The following observation will be used extensively.
\begin{observation}
  \label{obs:cheap-vs-small}
  If $U \subseteq V(G)$ is a large set, then there exists a cheap vertex in~$U$, or
  contrapositively: if a set $U \subseteq V(G)$ has only expensive vertices,
  then~$U$ is small.  Specifically it follows that in any yes instance $(G,k)$
  where~$F$ is a solution, there are at most~$2\sqrt k$ expensive vertices.
\end{observation}

This gives the following win-win situation: If a set~$X$ is small, then we can
``guess'' it, which means that we can in subexponential time enumerate all
candidates, and otherwise, we can guess a cheap vertex inside the set and its
``correct'' neighborhood.  In particular, since the set of expensive vertices is
small, we can guess it in the beginning.  For the remainder of the proof we will
assume that the graph~$G$ is a labeled graph, where some vertices are labeled as
cheap and others as expensive.  There will never be more than~$2\sqrt{k}$
vertices labeled expensive, however a vertex labeled expensive might very well
not be expensive in~$G$ and vice versa.  The idea is that we guess the expensive
vertices at the start of the algorithm and then bring this information along
when we recurse on subgraphs.

\subsubsection{Splitting Pairs and Unbreakable Segments}
\label{sec:splitt-pairs-unbr}

\begin{definition}[Splitting pair]
  Let~$G$ be a graph,~$k$ an integer,~$F$ a solution of $(G,k)$ and
  $(\mathcal{C}, \mathcal{I})$ a threshold decomposition of $G \triangle F$.  We
  then say that the vertices $u \in I_a$ and $v \in C_b$ is a \emph{splitting
    pair} if
  \begin{itemize}
  \item $a < b$,
  \item $u$ and $v$ are cheap,
  \item $\cup_{a < i < b} L_i$ consists of only expensive vertices.  Recall from
    Definition~\ref{def:threshold-partition} that~$L_i = C_i \cup I_i$.
  \end{itemize}
\end{definition}

\begin{definition}[Unbreakable]
  Let~$G$ be a graph,~$k$ an integer,~$F$ a solution of $(G,k)$ and
  $(\mathcal{C}, \mathcal{I})$ a threshold decomposition of $G \triangle F$.  We
  then say that a sequence of levels $(C_a, I_a), (C_{a+1}, I_{a+1}), \ldots,
  (C_b, I_b)$ is an \emph{unbreakable segment} if there is no splitting pair in
  the vertex set $\cup_{i \in [a, b]} (C_i \cup I_i)$.
  
  Furthermore, we say that an instance $(G,k)$ is \emph{unbreakable} if there
  exists an optimal solution~$F$ and a threshold decomposition $(\mathcal{C},
  \mathcal{I})$ of~$G \triangle F$ such that the entire decomposition is an
  unbreakable segment.  We also say that such a decomposition is a
  \emph{witness} of~$G$ being unbreakable.
\end{definition}

\begin{definition}
  Let~$G$ be a graph and $(\mathcal{C}, \mathcal{I})$ a threshold decomposition
  of~$G \triangle F$ for some solution~$F$.  Then we say that~$i$ is a
  \emph{transfer level} if
  \begin{itemize}
  \item for every $j > i$ it holds that~$C_j$ contains no cheap vertices and
  \item for every $j < i$ it holds that~$I_j$ contains no cheap vertices.
  \end{itemize}
\end{definition}

\begin{lemma}
  \label{lemma:unbreakable-graph-divind-level}
  Let $(G,k)$ be a yes instance of \STE{} with solution~$F$ such that~$G$ is
  unbreakable and $(\mathcal{C}, \mathcal{I})$ a witness.  Then there is a
  transfer level in $(\mathcal{C}, \mathcal{I})$.
\end{lemma}

\begin{proof}
  Suppose for a contradiction that the lemma is false.  Let~$a$ be maximal such
  that~$C_a$ contains a cheap vertex and~$b$ minimum such that~$I_b$ contains a
  cheap vertex.  Since~$i = a$ clearly satisfies the first condition, it must be
  the case that~$b < a$.  Increment~$b$ as long as~$b+1 < a$ and there is a
  cheap vertex in~$\cup_{i \in (b, a)} I_i$.  Then decrement~$a$ as long as~$b+1
  < a$ and there is a cheap vertex in~$\cup_{i \in (b, a)} C_i$.  Let~$u$ be a
  cheap vertex in~$C_a$ and~$v$ a cheap vertex in~$C_b$.  It follows from the
  procedure that they both exist.  Observe that~$u, v$ is indeed a splitting
  pair, which is a contradiction to~$G$ being unbreakable and $(\mathcal{C},
  \mathcal{I})$ being a witness.
\end{proof}

\begin{lemma}
  \label{lemma:unbreakable-graph-few-levels}
  Let $(G,k)$ be an instance of \STE~such that~$G$ is unbreakable and
  $(\mathcal{C}, \mathcal{I})$ a witness of this.  Then the number of levels in
  $(\mathcal{C}, \mathcal{I})$ is at most $2\sqrt{k} + 1$.
\end{lemma}
\begin{proof}
  Let~$i$ be the transfer level in $(\mathcal{C}, \mathcal{I})$.  It is
  guaranteed to exist by Lemma~\ref{lemma:unbreakable-graph-divind-level}.
  Observe that for every~$j > i$ it holds that~$C_i$ consists of expensive
  vertices and for every~$j < i$ it holds that~$I_i$ consists of expensive
  vertices.  It follows immediately that every level besides~$i$ contains at
  least one expensive vertex.  As there are at most~$2\sqrt{k}$ such vertices
  the result follows immediately.
\end{proof}

\begin{lemma}
  \label{lemma:unbreakable-graph-cheap-complete}  
  Let $(G,k)$ be an instance of \STE~such that~$G$ is unbreakable,
  $(\mathcal{C}, \mathcal{I})$ is a witness of this and~$F$ a corresponding
  solution.  If~$X$ is the set of cheap vertices in~$G$ then $(G \triangle
  F)[X]$ forms a complete split graph.
\end{lemma}
\begin{proof}
  Let~$t$ be the transfer level of the decomposition,~$u$ a cheap vertex
  in~$C_i$ and~$v$ a cheap vertex in~$I_j$ for some~$i$ and~$j$.  By the
  definition of~$t$ it holds that $i \leq t \leq j$.  It follows immediately
  that~$u$ and~$v$ are adjacent in $G \triangle F$ and the proof is complete.
\end{proof}

\newcommand{\unbreakAlg}{{\texttt{unbreakAlg}}}

We will now describe the algorithm \unbreakAlg.  It takes as input an instance
$(G, (C,I), k)$ of \STE, with the assumption that~$G$ is unbreakable and has
split partition $(C,I)$, and returns either an optimal solution~$F$ for~$(G, k)$
where~$|F| \leq k$ or correctly concludes that~$(G, k)$ is a no-instance.
Assume that~$(G, k)$ is a yes-instance.  Then there exists an optimal
solution~$F$ and a threshold decomposition~$(\mathcal{C}, \mathcal{I})$ of~$G
\triangle F$ that is a witness of~$G$ being unbreakable.
First, we guess the number of levels~$\ell$ in the decomposition, and by
Lemma~\ref{lemma:unbreakable-graph-few-levels}, we have that~$\ell \in [0,
2\sqrt{k}+1]$ and the transfer level $t \in [0, \ell]$.  Then we guess where the
at most~$2\sqrt{k}$ vertices that are expensive in~$G$ are positioned in
$(\mathcal{C}, \mathcal{I})$.  Observe that from this information we can obtain
all edges between expensive vertices in~$F$.
Finally, we put every cheap vertex in the level that minimizes the cost of
fixing its adjacencies into the expensive vertices while respecting that~$t$ is
the transfer level.  From this information we can obtain all adjacencies between
cheap and expensive vertices in~$F$.  Since the cheap vertices induces a
complete split graph, we reconstructed~$F$ and hence we return it.

\begin{lemma}
  \label{lemma:solving-unbreakable-segments}
  Given an instance $(G,k)$ of \STE~with~$G$ being unbreakable,
  \textup{\unbreakAlg{}} either gives an optimal solution or correctly concludes
  that $(G,k)$ is a no-instance in time $2^{O(\sqrt{k} \log{k})}$.
\end{lemma}
\begin{proof}
  Since the algorithm goes through every possible value for~$\ell$ and~$t$
  (according to Lemmata~\ref{lemma:unbreakable-graph-divind-level}
  and~\ref{lemma:unbreakable-graph-few-levels}), and every possible placement of
  the expensive vertices, the only thing remaining to ensure is that the cheap
  vertices are placed correctly.  However, since the cheap vertices form a
  complete split graph (according to
  Lemma~\ref{lemma:unbreakable-graph-cheap-complete}), the only cost associated
  with a cheap vertex is the number of expensive vertices in the opposite side
  it is adjacent to.  However, their placement is fixed, so we simply greedily
  minimize the cost of the vertex by putting it in a level that minimizes the
  number of necessary edits.
  
  If we get a solution from the above procedure, this solution is optimal.  On
  the other hand, if in every branch of the algorithm we are forced to edit more
  than~$k$ edges, then either $(G,k)$ is a no-instance, or~$G$ is not
  unbreakable.  Since the assumption of the algorithm is that~$G$ is
  unbreakable, we conclude that the algorithm is correct.
\end{proof}

\subsubsection{Divide and Conquer}
\label{sec:divide-conquer}

\newcommand{\solveAlg}{{\texttt{solveAlg}}}

We now explain the main algorithm.  The algorithm takes as input a graph~$G$,
together with a split partition $(C,I)$ and a budget~$k$.  In addition, it takes
a vertex set~$S$ which the algorithm is supposed to find an optimal solution
for.  The algorithm is recursive and either finds a splitting pair, in which it
recurses on a subset of~$S$, and if there is no splitting pair, then $G[S]$ is
unbreakable, and thus it simply runs \unbreakAlg{} on~$S$.
To avoid unnecessary recomputations, it uses memoization to solve already
computed inputs.

\medskip

The algorithm $\solveAlg(G, (C, I), k, S)$ returns an optimal solution for the
instance $(G[S],k)$, respecting the given split partition~$(C,I)$ in the
following manner:
\begin{enumerate}[(1)]
\item Run $\unbreakAlg(G[S],(C \cap S,I \cap S),k)$.
\item For every pair of cheap vertices $u \in I$ and $v \in C$, together with their
  correct neighborhoods~$N_u$ and~$N_v$, and every pair of subsets $C_X
  \subseteq C$ and $I_X \subseteq I$ of expensive vertices we do the following:
  Let $X = I_X \cup C_X$, $R_C = N_u$, $U_I = N_v \cap I$, $R_I = I \setminus (X \cup U_I)$ and
  $U_C = S \setminus (X \cup R_C \cup U_I \cup R_I)$.  Now, $U = U_I \cup U_C$
  is the unbreakable segment,~$X$ is the set of expensive vertices between the
  splitting pair, and $R = R_I \cup R_C$ is the remaining vertices.  We now
  \begin{enumerate}[(a)]
  \item Run $\unbreakAlg(G[U],(C \cap U,I \cap U),k)$ yielding a solution~$F_U$,
  \item solve $G[X]$ optimally by brute force since it has size at most $2 \sqrt
    k$, giving a solution~$F_X$, and
  \item recursively call $\solveAlg(G, (C, I), k, R)$ to solve the instance
    corresponding to the remaining vertices yielding~$F_R$.
  \end{enumerate}
  Finally we return~$F$, the union of~$F_U$,~$F_X$, and~$F_R$ together with all
  edges from $C \cap R$ and $I \cap (X \cup U)$, and all edges from $C \cap X$
  to $I \cap U$.
\end{enumerate}

\begin{figure}[t]
  \centering
  \begin{tikzpicture}[every node/.style={circle, draw, scale=.8}, scale=.4]
      \node (c1) at  (0,0)  {};
      \node (c2) at  (0,1)  {};
      \node (c3) at  (0,2)  {};
      \node (c4) at  (0,3)  {};
      \node (c5) at  (0,4)  {};
      \node (c6) at  (0,5)  {};
      \node (c7) at  (0,6)  {};
      \node[rectangle] (c8) at  (0,7)  {};
      \node (c9) at  (0,8)  {};
      \node (c10) at (0,9)  {};
      \node (c11) at (0,10) {};
      \node (c12) at (0,11) {};

      \node (i0) at (3,-1) {};
      \node (i1) at (3,0) {};
      \node (i2) at (3,1) {};
      \node (i3) at (3,2) {};
      \node[rectangle] (i4) at (3,3) {};
      \node (i5) at (3,4) {};
      \node (i6) at (3,5)  {};
      \node (i7) at (3,6)  {};
      \node (i8) at (3,7)  {};
      \node (i9) at (3,8)  {};
      \node (i10) at (3,9)  {};
      \node (i11) at (3,10) {};

      \node[draw=none] (CL) at (0, 12) {$C$};
      \node[draw=none] (IL) at (3, 12) {$I$};

      \draw (c3) -- (i4) -- (c4);
      \draw (c1) -- (i4) -- (c2);

      \draw (i9) -- (c8) -- (i10);
      \draw (i8) -- (c8) -- (i11);

      \draw [decorate,decoration={brace,mirror,amplitude=6pt}]
      (-0.5,11.3) -- (-0.5,6.7) node [draw=none,midway,xshift=-1cm] {$U$};

      \draw [decorate,decoration={brace,mirror,amplitude=4pt}]
      (-0.5,6.3) -- (-0.5,3.7) node [draw=none,midway,xshift=-1cm] {$X$};

      \draw [decorate,decoration={brace,mirror,amplitude=6pt}]
      (-0.5,3.3) -- (-0.5,-1.3) node [draw=none,midway,xshift=-1cm] {$R$};

    \end{tikzpicture}
    \caption{The partitioning of the vertex sets according to \solveAlg.  The
      square bags are the bags containing the splitting pair, $U$ is an
      unbreakable segment and the bags of $X$ contains exclusively expensive
      vertices.  The edges drawn indicates the neighborhoods of the splitting
      pair across the partitions.}
  \label{fig:solvealg}
\end{figure}
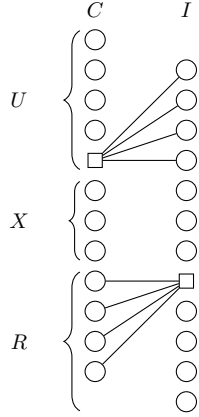

In \textit{(1)} we consider the option that there are no splitting pairs in~$G$.
In \textit{(2)} (see Figure~\ref{fig:solvealg}) we guess the uppermost splitting
pair in the partition and the neighborhood of these two vertices.  Then we guess
all of the expensive vertices that live in between the two levels of the
splitting pair.  Observe that these expensive vertices together with the
splitting pair partition the levels into three consecutive sequences.
The upper one,~$U$ is an unbreakable segment, the middle,~$X$ are the expensive
vertices and the lower one,~$R$ is simply the remaining graph.
When we apply \unbreakAlg{} on the upper part, brute force the middle one and
recurse with \solveAlg{} on the lower part, we get individual optimal solutions
for each three, finally we may merge the solutions and add all the remaining
edges (see end of \textit{(2)}).

\begin{lemma}
  \label{lemma:subept-correctness}
  Given a split graph $G=(V,E)$ with split partition $(C,I)$,
  \textup{\solveAlg{}} either returns an optimal solution for \STE{} on input
  $(G,(C,I),k,V)$, or correctly concludes that $(G,k)$ is a no-instance.
\end{lemma}
\begin{proof}
  If $(G,k)$ with split partition $(C,I)$ is a yes instance of \STE{} there is a
  solution~$F$ with threshold decomposition~$(\C,\I)$ and a sequence of pairs
  $(u_1,v_1), (u_2,v_2), \dots, (u_t,v_t)$ such that $u_1,v_1$ is the splitting
  pair highest in $(\C,\I)$, and $u_2,v_2$ in the highest splitting pair in the
  graph induced by the vertices in and below the level of~$v_1$, etc.  Since we
  in a state $(G, (C,I), k, S)$ try every possible pair of such cheap vertices
  and every possible neighborhood and set of expensive vertices, we exhaust all
  possibilities for any threshold editing of~$S$ of at most~$k$ edges.  Hence,
  if there is a solution, an optimal solution is returned.
  
  Thus, if ever an~$F$ is constructed of size $|F| > k$, we can safely conclude
  that there is no editing set $F^\star \subseteq C \times I$ of size at most~$k$ such that $G
  \triangle F^\star$ is a threshold graph.
\end{proof}

\begin{lemma}
  \label{lemma:subept-time}
  Given a split graph $G = (V,E)$ with split partition $(C,I)$ and an
  integer~$k$ with $|V(G)| = O(k^2)$, the algorithm \textup{\solveAlg{}}
  terminates in time $2^{O(\sqrt k \log k)}$ on input $(G,(C,I),k,V)$.
\end{lemma}
\begin{proof}
  By charging a set~$S$ for which \solveAlg{} is called with input
  $(G,(C,I),k,S)$ every operation except the recursive call, we need to
  \textit{(i)} show that there are at most $2^{O(\sqrt k \log k)}$ many sets $S
  \subseteq V$ for which \solveAlg{} is called, and \textit{(ii)} that the work
  done inside one such call is at most $2^{O(\sqrt k \log k)}$.
  
  For Case \textit{(i)}, we simply
  note that when \solveAlg{} is called with a set~$S$, the sets~$R$ on which we
  recurse are uniquely defined by $u,v,N_u,N_v,X$, and there are at most $O(k^4)
  \cdot 2^{O(\sqrt k \log k)^3} = 2^{O(\sqrt k \log k)}$ such configurations, so
  at most $2^{O(\sqrt k \log k)}$ sets are charged.
  Case \textit{(ii)} follows from the fact that we guess two vertices,~$u$
  and~$v$ and three sets,~$N_u$,~$N_v$ and~$X$.  For each choice we run
  \unbreakAlg{}, which runs in time $2^{O(\sqrt k \log k)}$ by
  Lemma~\ref{lemma:solving-unbreakable-segments}, and the brute force solution
  takes time $2^{O(\sqrt k \log(\sqrt k))}$.  The recursive call is charged to a
  smaller set, and merging the solutions into the final solution we return,~$F$,
  takes polynomial time.
  
  The two cases show that we charge at most $2^{O(\sqrt k \log k)}$ sets with
  $2^{O(\sqrt k \log k)}$ work, and hence \solveAlg{} completes after
  $2^{O(\sqrt k \log k)}$ steps.
\end{proof}

\bigskip

To conclude we observe that Theorem~\ref{thm:t-e-subept} follows directly from
the above exposition.  Given an input $(G,k)$ to \TE{}, from the previous
section we can in polynomial time obtain an equivalent instance with at most
$O(k^2)$ vertices.  Furthermore, by
Lemma~\ref{lemma:list-relevant-split-partitions} we may in time $2^{O(\sqrt k
  \log k)}$ time assume we are solving the problem \STE{}.  Finally, by
Lemmata~\ref{lemma:subept-correctness} and~\ref{lemma:subept-time}, the theorem
follows.

% subsection
\subsection{Editing to Chain Graphs}
\label{sec:subept-chain}

We finally describe which steps are needed to change the algorithm above into an
algorithm correctly solving \pname{Chain Editing} in subexponential time.

The main difference between \pname{Chain Editing} and \pname{Threshold Editing}
is that it is far from clear that the number of bipartitions is subexponential,
that is, is there a bipartite equivalent of the bound of the potential split
partitions as in Lemma~\ref{lemma:list-relevant-split-partitions}?
If we were able to enumerate all such ``potential bipartitions'' in
subexponential time, we could simply run a very similar algorithm to the one
above on the problem \pname{Bipartite Chain Editing}, where we are asked to
respect the bipartition (see Section~\ref{sec:np-hardness-chain} for the
definition of this problem).

It turns out that we indeed are able to enumerate all such potential
bipartitions within the allowed time:

\begin{lemma}
  \label{lem:potential-bipartitions}
  There is an algorithm which, given an instance $(G,k)$ for \pname{Chain
    Editing}, enumerates $\binom{|V|}{O(\sqrt k)} = 2^{O(\sqrt k \log |V|)}$
  bipartite graphs $H = (A,B,E')$ with $\left| E \triangle E' \right| \leq k$
  such that if $(G,k)$ is a yes instance, then one output $(H,k)$ will be a yes
  instance for \pname{Bipartite Chain Editing}, and furthermore is any yes
  instance $(H,k)$ is output, then $(G,k)$ is a yes instance.
  This also holds for the deletion and completion versions.
\end{lemma}
\begin{proof}
  We first mention that it is trivial to change the below proof into the proofs
  for the deletion and completion versions; One simply disallow one of the
  operations.  So we will prove only the editing version.  Furthermore, it is
  clear to see that if any output instance $(H,k)$ is a yes instance for
  \pname{Bipartite Chain Editing}, then $(G,k)$ was a yes instance for
  \pname{Chain Editing}.
  
  Consider any solution $H = (A,B,E')$ for an input instance $(G,k)$.  If either
  $\min\{|A|,|B|\} \leq 5\sqrt k$, then we can simply guess every such in
  subexponential time.  Hence, we assume that both sides of~$H$ are large.  But
  this means, by Observation~\ref{obs:cheap-vs-small}, that both $A$ and $B$
  have cheap vertices.  Let $v_A$ be a cheap vertex as low as possible in $A$
  and $v_B$ be a cheap vertex as high as possible in $B$.  It immediately
  follows from the same observation that the set of vertices below $v_A$, $A_X$
  is a set of expensive vertices, and the same for the vertices above $v_B$,
  $B_X$.  Since $v_A$ and $v_B$, we know that we can in subexponential time
  correctly guess their neighborhoods in $H$ and we can similarly guess $A_X$
  and $B_X$.
  
  Now, since we know $v_A$, $v_B$, $N_H(v_A)$ and $N_H(v_B)$, as well as $A_X$
  and $B_X$, the only vertices we do now know where to place, are the vertices
  in $A$ which are in the levels above $\lev(v_B)$, call them $A_Y$, and the
  vertices in $b$ which are in the levels below $\lev(v_A)$.  However, we know
  which set this is, that is, we know $Z = A_Y \cup B_Y$.  Define now $A_M = A
  \setminus ( A_Y \cup A_X \cup \{v_A\})$ and similarly $B_M = B \setminus ( B_Y
  \cup B_X \cup \{v_B\})$.  These are the vertices living in the middle of $A$
  and $B$, respectively.
  
  We now know that the vertices of $Z$ should form an independent set.  This
  follows from the fact that $A_M$ and $B_M$ are both non-empty.  Hence, the
  vertices of $A_Y$ are in higher levels than all of $B_Y$, and since there are
  no edges going from a vertex in $A$ to a vertex lower in $B$, and each of $A$
  and $B$ are independent sets, $Z$ must be an independent set.
  
  The following is the crucial last step.  We can in subexponential time guess
  the partitioning of levels of both $A_X$ and of $B_X$, since they are both of
  sizes at most $2 \sqrt k$.  When knowing these levels, we can greedily insert
  each vertex in $Z$ into either $A$ and $B$ by pointwise minimizing the cost; A
  vertex $z \in Z$ can safely be places in the level of $A$ or $B$ which
  minimizes the cost of making it adjacent to only the vertices of $B_X$ above
  its level, or by making it adjacent to only the vertices below its level in
  $A_X$.
\end{proof}

Given the above lemma, we may work on the more restricted problem,
\pname{Bipartite Chain Editing}.  The rest of the algorithm actually goes
through without any noticeable changes:
\begin{theorem}
  \label{thm:chain-editing-subept}
  \pname{Chain Editing} is solvable in time~$2^{O(\sqrt k \log k)} + \poly(n)$.
\end{theorem}
\begin{proof}
  On input $(G,k)$ we first run the kernelization algorithm from
  Section~\ref{sec:adapt-kern-modif}, and then we enumerate every potential
  bipartition according to Lemma~\ref{lem:potential-bipartitions}.  Now, for
  each bipartition $(A,B)$ we make $A$ into a clique, and run the \STE{}
  algorithm from Section~\ref{sec:subept-threshold} (see also
  Proposition~\ref{prop:chain-defs}).

  Now, $(G,k)$ is a yes instance if and only if there is a bipartition $(A,B)$
  such that when making $A$ into a clique, the resulting instance is a yes
  instance for \STE.
\end{proof}

\begin{corollary}
  \pname{Chain Deletion} and \pname{Chain Completion} are solvable in
  time~$2^{O(\sqrt k \log k)} + \poly(n)$.
\end{corollary}

% Section 6
\section{Conclusion}
\label{sec:conclusion}

In this paper we showed that the problems of editing edges to obtain a threshold
graph and editing edges to obtain a chain graph are \NP-complete.  The latter
solves a conjecture in the positive from Natanzon et
al.~\cite{natanzon2001complexity} and both results answer open questions from
Sharan~\cite{sharan2002graph}, Burzyn et al.~\cite{burzyn2006np}, and
Mancini~\cite{mancini2008graph}.

On the positive side, we show that both \TE{} and \CE{} admit quadratic kernels,
i.e., given a graph $(G,k)$, we can in polynomial time find an equivalent
instance $(G',k)$ where $|V(G')| = O(k^2)$, and furthermore,~$G'$ is an induced
subgraph of~$G$.
We also show that these results hold for the deletion and completion variants as
well, and these results answer open questions by Liu et al.\ in a recent survey
on kernelization complexity of graph modification
problems~\cite{liu2014overview}.

Finally we show that both problems admit subexponential algorithms of time
complexity $2^{O(\sqrt{k} \log k)} + \poly(n)$.  This answers a recent open
question by Liu et al.~\cite{liu2015edge}.

\bigskip

In addition, we give a proof for the \NP-hardness of \pname{Chordal Editing}
which has been announced several places but which the authors have been unable
to find.
However, our \NP-completeness proof for \pname{Chordal Editing} suffers a
quadratic blow-up from \pname{3Sat}, i.e., $k = \Theta (|\phi|^2)$, so we cannot
get better than $2^{o(\sqrt k)} \cdot \poly(n)$ lower bounds from this
technique.  The current best algorithm for \pname{Chordal
  Editing}\footnote{Here, the authors take \pname{Chordal Editing} to allow
  vertex deletions.} runs in time $2^{O(k \log k)} \cdot
\poly(n)$~\cite{cao2014chordal}, and so this leaves a big gap.  It would be
interesting to see if we can achieve tighter lower bounds, e.g., $2^{o(k)} \cdot
\poly(n)$ time lower bounds for \pname{Chordal Editing} assuming ETH together
with a $2^{O(k)} \cdot \poly(n)$ time algorithm.

\bigskip

\paragraph{Acknowledgment.}  The authors would like to express their gratitude
to Ulrik Brandes and Mehwish Nasim for helpful comments on an early draft of
this paper.

\medskip\noindent
The research leading to these results has received funding from the Research
Council of Norway, Bergen Research Foundation under the project Beating Hardness
by Preprocessing and the European Research Council under the European Union's
Seventh Framework Programme (FP/2007-2013) / ERC Grant Agreement n.~267959.

\smallskip\noindent
Blair D.\ Sullivan supported in part by the Gordon \& Betty Moore Foundation as
a DDD Investigator and the DARPA GRAPHS program under SPAWAR Grant
N66001-14-1-4063.  Any opinions, findings, and conclusions or recommendations
expressed in this publication are those of the author(s) and do not necessarily
reflect the views of DARPA, SSC Pacific, or the Moore Foundation.

\end{document}